\documentclass[10pt,twocolumn,twoside]{IEEEtran}
\usepackage[latin9]{inputenc}
\usepackage[active]{srcltx}
\usepackage{color}
\usepackage{float}
\usepackage{bm}
\usepackage{amsmath}
\usepackage{amssymb}

\makeatletter

\floatstyle{ruled}
\newfloat{algorithm}{tbp}{loa}
\providecommand{\algorithmname}{Algorithm}
\floatname{algorithm}{\protect\algorithmname}

\usepackage[noadjust]{cite}

\usepackage[T1]{fontenc}
\usepackage{float}
\usepackage{subfigure}
\makeatletter

\floatstyle{ruled}
\newfloat{algorithm}{tbp}{loa}
\providecommand{\algorithmname}{Algorithm}
\floatname{algorithm}{\protect\algorithmname}

\newtheorem{theorem}{Theorem}

\newtheorem{lemma}[theorem]{Lemma}

\newtheorem{definition}[theorem]{Definition}

\newtheorem{remark}[theorem]{Remark}

\newcounter{algo}

\ifCLASSINFOpdf
\else
\fi
\renewcommand{\Re}{\mathbb{R}}

\newcommand{\trt}{^{\scriptscriptstyle T}}
\newcommand{\bx}{{\mathbf{x}}}
\newcommand{\by}{{\mathbf{y}}}
\newcommand{\bz}{{\mathbf{z}}}

\newcommand{\argmin}{\mathop{\rm argmin}}

\usepackage[footnotesize]{caption}
 \usepackage{flushend}

\makeatother

\begin{document}

\title{\vspace{-0.2cm} Parallel and  Distributed Methods for   Nonconvex 
Optimization$-$Part I: Theory}

\author{Gesualdo Scutari, Francisco Facchinei, Lorenzo Lampariello, and Peiran
Song\vspace{-0.7cm}%
  \thanks{Scutari is with the School of Industrial Engineering and the Cyber Center (Discovery Park), Purdue University, West-Lafayette, IN, USA; email: \texttt{gscutari@purdue.edu}. Facchinei and Lampariello are with the Dept. of Computer, Control, and Management Engineering, University of Rome ``La Sapienza'',  Rome, Italy; emails: \texttt{<facchinei, lampariello>@diag.uniroma1.it}. Song is with  the School of Information and Communication Engineering, Beijing Information Science and Technology University, Beijing,  China; email: \texttt{peiransong@bistu.edu.cn}.  \newline The work of  Scutari was supported by the USA National Science Foundation under Grants CMS 1218717, CCF 1564044, and CAREER Award 1254739. The work of Facchinei was partially supported by the MIUR project PLATINO, Grant  n. PON01$\_$01007.  Part of this work appeared in \cite{Scutari_ICASSP14,ScuFaccLampSonICASSP16} and on arxiv \cite{FacchineiLamparielloScutariTSP14-PartI} on Oct. 2014.}}

\maketitle
\begin{abstract}
In this two-part paper, we propose a general algorithmic framework
for the minimization of a nonconvex smooth function subject to \emph{nonconvex}
smooth constraints. The algorithm solves a sequence of \emph{(separable)
strongly convex} problems and mantains feasibility at each iteration. Convergence to a stationary solution of
the original nonconvex optimization is established. Our framework
is very general and flexible and  unifies several existing Successive
Convex Approximation (SCA)-based algorithms 
More importantly, and differently from current SCA approaches,
it naturally leads to \emph{distributed and parallelizable} implementations
for a large class of nonconvex problems.

This Part I is devoted to the description of the framework in its
generality. In Part II we customize our general methods to several
multi-agent optimization problems, mainly in communications and networking;
the result is a new class of centralized and \emph{distributed} algorithms that compare
favorably to existing ad-hoc (centralized) schemes.
\vspace{-0.3cm}

\end{abstract}

\IEEEpeerreviewmaketitle

\section{Introduction\label{sec:intro} }

\IEEEPARstart{T}{he} minimization of a nonconvex (smooth) objective
function $U:\,\mathcal{K}\to\mathbb{R}$ subject to  convex constraints
$\mathcal{K}$ and nonconvex ones $g_{j}(\bx)\leq0$, with $g_{j}:\,\mathcal{K}\to\mathbb{R}$
smooth,
\vskip-0.1cm
\begin{equation}
\begin{array}{cl}
\underset{\bx}{\textnormal{min}} & \,\, U(\bx)\smallskip\\
\mbox{s.t.} & \vspace{-0.4cm}\\
 & \hspace{-0.1cm}\left.\begin{array}{l}
g_{j}(\bx)\le0,\; j=\,1,\ldots,m\\[5pt]
\bx\in\mathcal{K}
\end{array}\right\} \triangleq\mathcal{X},
\end{array}\tag{$\mathcal P$}\label{eq:nonnonU}
\end{equation}
\vskip-0.1cm

\noindent
is an ubiquitous problem that arises in many fields, ranging from
signal processing to communication, networking, machine learning,
etc.

It is hardly possible here to even summarize the huge amount of solution
methods that have been proposed for problem \eqref{eq:nonnonU}. Our
focus in this paper is on \emph{distributed} algorithms converging
to stationary solutions of \ref{eq:nonnonU} while \emph{ preserving
the feasibility of the iterates}. While the former feature needs no
further comment, the latter is motivated by several reasons. First,
in many cases the objective function $U$ is not even defined outside
the feasible set; second, in some applications one may have to interrupt
calculations before a solution has been reached and it is then important
that the current iterate is feasible; and third, in on-line implementations
it is mandatory that some constraints are satisfied by every iterate
(e.g., think of power budget or interference constraints). As far
as we are aware of, there exists no method for the solution of \ref{eq:nonnonU}
in its full generality that is both feasible {\em and} distributed.

Existing efforts pursuing the above design criteria include: 1) Feasible Interior Point (FIP) methods (e.g., 
\cite{byrd2003feasible,FiaMcC68}), 2) Feasible Sequential Quadratic
Programming (FSQP) methods (e.g., \cite{LawTit01}); 3) Parallel Variable
Distribution (PVD) schemes (e.g., \cite{FerrMan94,SagSol02,Sol98});
 4) SCA algorithms (in the spirit of \cite{MarWrig78,BeckBenTTetr10,RazaviyaynHongLuo_SIAMOPT13,ScutFacchSonPalPang13,AlvScutPang13,QuDie11}); and
some specialized algorithms with roots in the structural optimization field (e.g., \cite{fleury1989conlin,svanberg1987method,svanberg2002class}).
FIP and FSQP methods  maintain feasibility throughout the
iterations but are centralized and computationally expensive. PVD
schemes are suitable for implementation over parallel architectures
but they require an amount of information exchange/knowledge that
is often not compatible with a distributed architecture (for example they
cannot be applied to the case study discussed in Part II of the paper
\cite{FacchineiLamparielloScutariTSP14-PartII}). Furthermore, when
applied to problem \ref{eq:nonnonU}, they call for the solution
of possibly difficult nonconvex (smaller) subproblems; and convergence
has been established only for convex \cite{FerrMan94,Sol98} or nonconvex
but block separable $g_{j}$s \cite{SagSol02}. Standard
SCA methods are centralized \cite{MarWrig78,BeckBenTTetr10,QuDie11},
with the exception of \cite{ScutFacchSonPalPang13,AlvScutPang13}
and some instances of \cite{RazaviyaynHongLuo_SIAMOPT13} that lead
instead to distributed schemes. However, convergence conditions have
been established only in the case of \emph{strongly convex} \emph{$U$}
\cite{BeckBenTTetr10} or \emph{convex} and separable $g_{j}$s \cite{RazaviyaynHongLuo_SIAMOPT13,ScutFacchSonPalPang13,AlvScutPang13}.
Finally, methods developed in the structural engineering  field, including  \cite{fleury1989conlin,svanberg1987method,svanberg2002class},  share some similarities with our approach, but in most cases
they lack reliable mathematical foundations or do not prove convergence to stationary points of the original problem \ref{eq:nonnonU}.
We refer to Sec. \ref{sub:Related-works} for a more detailed discussion
on existing works.

In this paper we propose a new  framework for
the general formulation \ref{eq:nonnonU} which, on one hand, maintains
feasibility and, on the other hand, leads, under very mild additional
assumptions, to parallel and distributed solution methods. The essential,
natural idea underlying the proposed approach is to compute a solution of
\ref{eq:nonnonU} by solving  a sequence of (simpler) strongly convex subproblems
whereby the nonconvex objective function and constraints are replaced
by suitable convex approximations; the subproblems can be then solved
(under some mild assumptions) in a distributed fashion using standard
primal/dual decomposition techniques (e.g., \cite{Palomar-Chiang_ACTran07-Num,Bertsekas_Book-Parallel-Comp}).
Additional key features of the proposed method are: i) it includes
as special cases several classical SCA-based algorithms, such as (proximal)
gradient or Newton type methods, block coordinate (parallel) descent
schemes, Difference of Convex (DC) functions approaches, convex-concave
approximation methods; ii) our convergence conditions unify and extend
to the general class \ref{eq:nonnonU} those of current (centralized)
SCA methods; iii) it offers much flexibility in the choice of
the convex approximation functions: for instance, as a major departure
from current SCA-based methods {[}applicable to special cases of \ref{eq:nonnonU}{]}
\cite{MarWrig78,RazaviyaynHongLuo_SIAMOPT13} and DC programs \cite{QuDie11},
the proposed approximation of the objective function $U$ need not
be a tight \emph{global upper bound }of $U$, a fact that significantly
enlarges the range of applicability of our framework;
and iv) by allowing alternative choices for the convex approximants,
it encompasses a gamut of novel algorithms, offering great flexibility
to control iteration complexity, communication overhead and convergence
speed, and all converging under the \emph{same} conditions. Quite
interestingly, the proposed scheme leads to new efficient
algorithms even when customized to solve well-researched problems, including power
control problems in cellular systems \cite{DC-BranchAndBound-1,DC-BranchAndBound-5,DC-Linearization-1,ChiTanPalONJul07},
MIMO relay optimization \cite{DC-Polynomial}, dynamic spectrum management
in DSL systems \cite{DC-BranchAndBound-2,TsiaflakisMoonenTSP08},
sum-rate maximization, proportional-fairness and max-min optimization
of SISO/MISO/MIMO ad-hoc networks \cite{ScutFacchSonPalPang13,SchmidtShiBerryHonigUtschick-SPMag, MochaourabCaoJorswieck_RateProfile_arxiv13,QiuZhangLuoCui_maxminTSP11},
robust optimization of CR networks \cite{KimGiannakisIT11,DallAnese2012,YangScutariPalomar_JSAC13},
transmit beamforming design for multiple co-channel multicast groups
\cite{Sidiropoulos:2006je,Karipidis:2008bn}, and cross-layer design
of wireless networks \cite{ConvexSumSeparable-2,Hande:2008ty}.  
Part II of the paper \cite{FacchineiLamparielloScutariTSP14-PartII} is devoted to the application of the proposed algorithmic framework to some of the aforementioned problems (and their generalizations).
Numerical results show that our schemes compare favorably to existing
ad-hoc ones (when they exist).

The rest of this two-part paper is organized as follows. Sec. \ref{sec:Technical-preliminaries}
introduces the main assumptions underlying the study of the optimization
problem \ref{eq:nonnonU} and provides an informal description of
our new algorithms. Sec. \ref{sec:Algorithmic-framework} presents
our novel framework based on SCA, whereas Sec. \ref{sec:Distributed-implementation}
focuses on its distributed implementation in the primal and dual domain.
Finally, Sec.\ref{sec:Conclusions} draws some conclusions. 
In Part II of
the paper \cite{FacchineiLamparielloScutariTSP14-PartII} we apply
our algorithmic framework to several resource allocation problems
in wireless networks and provide extensive numerical results showing
that the proposed algorithms compare favorably to state-of-the-art
schemes.\vspace{-0.3cm}

\section{Technical preliminaries and main idea\label{sec:Technical-preliminaries}}

In this section we introduce the main assumptions underlying the
study of the optimization problem \ref{eq:nonnonU} along with some
technical results that  will be instrumental to describe our approach.
We also provide an informal description of our new algorithms
that sheds light on the core idea of the proposed  decomposition
technique. The formal description of the framework is given in Sec.
\ref{sec:Algorithmic-framework}.

Consider problem \ref{eq:nonnonU}, whose feasible set is denoted
by $\mathcal{X}$.\smallskip{}

\noindent \textbf{Assumption 1.} We make the blanket assumptions:

\noindent A1) $\mathcal{K}\subseteq\mathbb{R}^{n}$ is closed and
convex (and nonempty);

\noindent A2) $U$ and each $g_{j}$ are continuously differentiable
on $\mathcal{K}$;

\noindent A3) $\nabla_{\bx}U$ is Lipschitz continuous on $\mathcal{K}$
with constant $L_{\nabla U}$.

\noindent A4) $U$ is coercive on $\mathcal{K}$.\smallskip{}

The assumptions above are quite standard and are satisfied by a large
class of problems of practical interest. In particular,
A4 guarantees that the problem has a solution, even when the
feasible set $\mathcal{X}$ is not bounded. Note that we do not assume
convexity of $U$ and $g_{1},\ldots,g_{m}$; without loss of generality, convex constraints,
if present, are accommodated in the set $\mathcal{K}$.

Our goal is to efficiently compute locally optimal solutions of  \ref{eq:nonnonU},
possibly in a distributed way, while preserving the feasibility of
the iterates. Building on the idea of SCA methods, our approach
consists in solving a sequence of \emph{strongly convex inner} approximations
of \ref{eq:nonnonU} in the form: given   $\bx^{\nu}\in\mathcal{X}$
\begin{equation}
 \begin{array}{cl}
\underset{\bx}{\textnormal{min}} & \,\,\tilde{U}(\bx;\bx^{\nu})\medskip\\
\mbox{s.t.} & \vspace{-0.4cm}\\
 & \hspace{-0.1cm}\left.\begin{array}{l}
\tilde{g}_{j}(\bx;\bx^{\nu})\le0,\; j=\,1,\ldots,m\\[5pt]
\bx\in\mathcal{K}
\end{array}\right\} \triangleq\mathcal{X}(\bx^{\nu}),
\end{array}\tag{$\mathcal P_{\bx^\nu}$}
\label{eq:k2}
\end{equation}
where $\tilde{U}(\bx;\bx^{\nu})$ and $\tilde{g}_{j}(\bx;\bx^{\nu})$
represent approximations of $U(\bx)$ and $g_{j}(\bx)$ at the current
iterate $\bx^{\nu}$, respectively, and $\mathcal{X}(\bx^{\nu})$
denotes the feasible set of \ref{eq:k2}.

We introduce next a number of assumptions that will be used throughout
the paper.

\noindent \textbf{Assumption 2 (On $\tilde{U}$).} Let $\tilde{U}:\mathcal{K}\times\mathcal{X}\rightarrow\mathbb{R}$
be a function continuously differentiable with respect to the first argument and such that:

\noindent B1)$\,\tilde{U}(\bullet;\by)$ is uniformly strongly convex
on $\mathcal{K}$ with constant $c_{\tilde{U}}\!\!>\!0$, i.e. $\forall \bx, \bz \in  \mathcal{K}, \; \forall \by \in \mathcal{X}$
$$
(\bx - \bz)^T \left(\nabla_{\bx}\tilde{U}(\bx;\by)- \nabla_{\bx}\tilde{U}(\bz;\by)\right)\geq
c_{\tilde{U}} \|\bx -\bz\|^2;
$$

\noindent B2) $\nabla_{\bx}\tilde{U}(\by;\by)=\nabla_{\bx}U(\by),$
for all $\by\in\mathcal{X}$;

\noindent B3) $\nabla_{\bx}\tilde{U}(\bullet;\bullet)$ is continuous
on $\mathcal{K}\times\mathcal{X}$;

\noindent where $\nabla_{\bx}\tilde{U}(\mathbf{u};\mathbf{w})$ denotes
the partial gradient of $\tilde{U}$ with respect to the first argument
evaluated at $(\mathbf{u};\mathbf{w})$.\smallskip{}

\noindent \textbf{Assumption 3 (On $\tilde{g}_{j}$s).} Let each
$\tilde{g}_{j}:\,\mathcal{K}\times\mathcal{X}\,\to\,\mathbb{R}$ satisfy
the following:

\noindent C1) $\tilde{g}_{j}(\bullet;\by)$ is convex on $\mathcal{K}$
for all $\by\in\mathcal{X}$;

\noindent C2) $\tilde{g}_{j}(\by;\by)=g_{j}(\by)$, for all $\by\in\mathcal{X}$;

\noindent C3) $g_{j}(\bx)\le\tilde{g}_{j}(\bx;\by)$ for all $\bx\in\mathcal{K}$
and $\by\in\mathcal{X}$;

\noindent C4) $\tilde{g}_{j}(\bullet;\bullet)$ is continuous on $\mathcal{K}\,\times\,\mathcal{X}$;

\noindent C5) $\nabla_{\bx}g_{j}(\by)=\nabla_{\bx}\tilde{g}_{j}(\by;\by)$,
for all $\by\in\mathcal{X}$;

\noindent C6) $\nabla_{\bx}\tilde{g}_{j}(\bullet;\bullet)$ is continuous
on $\mathcal{K}\times\,\mathcal{X}$;

\noindent where $\nabla_{\bx}\tilde{g}_{j}(\by;\by)$ denotes the
(partial) gradient of $\tilde{g}_{j}$ with respect to the first argument
evaluated at $\by$ (the second argument is kept fixed at $\by$).\smallskip

For some results we need stronger continuity properties of the (gradient
of the) approximation functions.

\noindent \textbf{Assumption 4}

\noindent B4) $\nabla_{\bx}\tilde{U}(\mathbf{x};\bullet)$ is uniformly
Lipschitz continuous on $\mathcal{X}$ with constant $\tilde{L}_{\nabla,2}$;

\noindent B5) $\nabla_{\bx}\tilde{U}(\bullet;\mathbf{y})$ is uniformly
Lipschitz continuous on $\mathcal{K}$ with constant $\tilde{L}_{\nabla,1}$;

\noindent C7) Each $\tilde{g}_{j}(\bullet;\bullet)$ is Lipschitz
continuous on $\mathcal{K}\times\mathcal{X}$.\smallskip{}

The key assumptions are B1, C1, and C3: B1 and C1 make  \ref{eq:k2} 
strongly convex, whereas C3 guarantees $\mathcal{X}(\bx^{\nu})\subseteq\mathcal{X}$
(iterate feasibility). The others are technical conditions (easy to
be satisfied in practice) ensuring that the approximations have
the same local first order behavior of the original functions. 
In the next section we provide some examples of approximate functions
that automatically satisfy Assumptions 2-4. As a final remark,
we point out that Assumptions 1-3 are in many ways similar \emph{but
generally weaker} than those used in the literature in order to solve special
cases of problem \ref{eq:nonnonU} \cite{MarWrig78,BeckBenTTetr10,RazaviyaynHongLuo_SIAMOPT13,ScutFacchSonPalPang13,AlvScutPang13}.
For instance, \cite{RazaviyaynHongLuo_SIAMOPT13,ScutFacchSonPalPang13,AlvScutPang13}
studied the simpler case of convex constraints; moreover, \cite{RazaviyaynHongLuo_SIAMOPT13}
requires the convex approximation $\tilde{U}(\bullet;\bx^{\nu})$
to be a \emph{global upper bound} of the nonconvex objective function
$U(\bullet)$, while we do not. The upper bound condition
C3 is assumed also in \cite{MarWrig78,BeckBenTTetr10} but, differently
from those works, we are able to handle also nonconvex objective functions
(rather than only strongly convex ones). Our weaker conditions on
the approximations $\tilde{U}$ and $\tilde{g}$ along with a more
general setting allow us to deal with a much larger class of problems
than \cite{RazaviyaynHongLuo_SIAMOPT13,AlvScutPang13,ScutFacchSonPalPang13,MarWrig78,BeckBenTTetr10};
see Part II of the paper \cite{FacchineiLamparielloScutariTSP14-PartII}
for specific examples. \vspace{-0.2cm}

\subsection{\noindent Regularity conditions\label{sub:Regularity-conditions}}

We conclude this section mentioning certain standard regularity
conditions on the stationary points of constrained optimization problems.
These conditions are needed in the study of the convergence properties
of our method. \vspace{-0.2cm}

\noindent \begin{definition}[Regularity] A point $\bar{\bx}\in\mathcal{X}$
 is called \emph{regular} for \ref{eq:nonnonU}
if the Mangasarian-Fromovitz Constraint Qualification (MFCQ) holds
at $\bar{\bx}$, that is (see e.g. \cite[Theorem 6.14]{RockWets98}) if the following
implication is satisfied:
\begin{equation}
\left.\begin{array}{c}
\mathbf{0}\in\sum_{j\in\bar{J}}\mu_{j}\nabla_{\mathbf{x}}g_{j}(\bar{\mathbf{x}})+N_{\mathcal{K}}(\bar{\mathbf{x}})\\[5pt]
\mu_{j}\ge0,\;\forall j\,\in\,\bar{J}
\end{array} \hspace{-5pt} \right\} \Rightarrow\mu_{j}=0,\,\forall j\,\in\,\bar{J},\label{eq:MFCQ-2}
\end{equation}
where $N_{\mathcal{K}}(\bar{\mathbf{x}})\triangleq\{\mathbf{d}\in\mathcal{K}\,:\,\mathbf{d}^{T}(\mathbf{y}-\mathbf{\bar{\mathbf{x}}})\leq0,\,\,\forall\mathbf{y}\in\mathcal{K}\}$
is the normal cone to $\mathcal{K}$ at $\bar{\mathbf{x}}$, and $\bar{J}\triangleq\{j\in\{1,\ldots,m\}:\, g_{j}(\bar{\mathbf{x}})=0\}$
is the index set  of those (nonconvex) constraints that are active
at $\bar{\mathbf{x}}$.

A similar definition holds for problem
 \ref{eq:k2}: a point $\bar{\bx}\in\mathcal{X}(\bx^\nu)$  is called \emph{regular}
for  \ref{eq:k2} if
\begin{equation}
\left.\begin{array}{c}
\mathbf{0}\in\sum_{j\in\bar{J}}\mu_{j}\nabla_{\mathbf{x}}\tilde g_{j}(\bar{\mathbf{x}}; \bx^\nu)+N_{\mathcal{K}}(\bar{\mathbf{x}})\\[5pt]
\mu_{j}\ge0,\;\forall j\,\in\,\bar{J}^\nu
\end{array} \hspace{-5pt} \right\} \Rightarrow\mu_{j}=0,\,\forall j\,\in\,\bar{J}^\nu\hspace{-2pt},\label{eq:MFCQ-2}
\end{equation}
where  $\bar{J}^\nu\triangleq\{j\in\{1,\ldots,m\}:\, \tilde g_{j}(\bar{\mathbf{x}}; \bx^\nu)=0\}$.
 \hfill $\square$\end{definition}\smallskip{}

\indent We point out that the regularity of $\bar{\bx}$ is implied by
 stronger but easier to be checked CQs, such as the
Linear Independence CQ, see \cite[Sec. 3.2]{FacchPang06} for more
details. Note that if the feasible set is convex, as it is in  \ref{eq:k2},
the MFCQ is equivalent to the Slater's CQ; for a set like
$\mathcal{X}(\mathbf{x}^{\nu})$, Slater's CQ reads
\[
\text{ri}(\mathcal{K})\cap \mathcal{X}_{g}^{<}(\mathbf{x}^{\nu})\neq\emptyset,
\]
where $\mathcal{X}_{g}^{<}(\mathbf{x}^{\nu})\triangleq\{\mathbf{x}\in\mathcal{K}:\,\tilde{g}_{j}(\mathbf{x};\mathbf{\mathbf{x}^{\nu}})<0,\, j=\,1,\ldots,m\}$ and $\text{ri}(\mathcal{K})$ is the relative interior of $\mathcal{K}$
(see, e.g., \cite[Sec. 1.4]{BertsekasNedicOzdaglar_book_convex03}). In particular, this means that for problem  \ref{eq:k2} either the MFCQ holds at all the feasible points or it does not hold at any point.
Furthermore, because of C2 and C5, a point $\bar{\bx}$ is regular for \ref{eq:nonnonU} if and only if $\bar{\bx}$ is regular for $\mathcal P_{\bar{\bx}}$  (and, therefore, if  any feasible point of $\mathcal P_{\bar{\bx}}$ is regular).

We recall that $\bar\bx$ is a stationary point of problem \ref{eq:nonnonU}, if
$$
\begin{array}{c}
\mathbf{0}\in \nabla_{\mathbf{x}} U(\bar\bx) + \sum_{j\in\bar{J}}\mu_{j}\nabla_{\mathbf{x}}g_{j}(\bar{\mathbf{x}})+N_{\mathcal{K}}(\bar{\mathbf{x}})\\[5pt]
\mu_{j}\ge0,\;\forall j\,\in\,\bar{J}
\end{array}
$$
for some suitable Lagrange multipliers $\mu_j$s. It is well-known that a regular (local) minimum point of problem \ref{eq:nonnonU} is also stationary. Finding stationary points is actually the classical goal of solution algorithms for nonconvex problems.

In order to simplify the presentation, in  the rest of this paper we  assume the following regularity condition.
\smallskip

\noindent \textbf{Assumption 5}
All  feasible points of problem \ref{eq:nonnonU} are regular.\smallskip

 \noindent One could relax  this assumption and require regularity only at specific points, but at the cost of more convoluted statements; we leave this task to the reader. We remark,  once again, that Assumption 5 implies that
 any feasible point of $\mathcal P_{\bar{\bx}}$ is regular.

\vspace{-0.2cm}

\section{Algorithmic framework\label{sec:Algorithmic-framework}}

We are now ready to formally introduce the proposed solution method for \ref{eq:nonnonU}.
Note first that,
because of B1 and C1, each subproblem  \ref{eq:k2} is strongly
convex and thus has a unique solution, which is denoted by $\hat{\bx}(\bx^{\nu})$
(a function of $\bx^{\nu}$):
\begin{equation}
\hat{\bx}(\bx^{\nu})\triangleq\argmin_{\bx\in\mathcal{X}(\bx^{\nu})}\tilde{U}(\bx;\bx^{\nu}).\label{eq:best-response}
\end{equation}

The proposed convex approximation  method consists in solving iteratively the optimization
problems (\ref{eq:best-response}), possibly including a step-size
in the iterates; we named it  iNner cOnVex Approximation (NOVA) algorithm. The formal description of the NOVA algorithm along with
its convergence properties are given in Algorithm \ref{algoC} and
Theorem \ref{th:conver}, respectively.

\begin{algorithm}[H]
\textbf{Data}: $\gamma^{\nu}\in(0,1]$, $\bx^{0}\,\in\,\mathcal{X}$;
set $\nu=0$.

(\texttt{S.1}) If $\bx^{\nu}$ is a stationary solution of \ref{eq:nonnonU}:
\texttt{STOP}.

(\texttt{S.2}) Compute
$\hat{\bx}(\bx^{\nu}),$
the solution  of  \ref{eq:k2} {[}cf. \eqref{eq:best-response}{]}.

(\texttt{S.3}) Set $\bx^{\nu+1}=\bx^{\nu}+\gamma^{\nu}(\hat{\bx}(\bx^{\nu})-\bx^{\nu})$.

(\texttt{S.4}) $\nu\leftarrow\nu+1$ and go to step (\texttt{S.1}).

\protect\caption{\hspace{-3pt}\textbf{: }\label{algoC} NOVA Algorithm for \ref{eq:nonnonU}. }
\end{algorithm}
\begin{theorem}\label{th:conver} Given the nonconvex problem \ref{eq:nonnonU}
under Assumptions 1-3 and 5, let $\{\bx^{\nu}\}$ be the sequence generated
by Algorithm 1. The following hold.

\noindent (a) $\mathbf{x}^{\nu}\in\mathcal{X}(\mathbf{x}^{\nu})\subseteq\mathcal{X}$
for all $\nu\geq 0$ (iterate feasibility);

\noindent (b) If the step-size $\gamma^{\nu}$ and $c_{\tilde{U}}$
are chosen so that
\begin{equation}
0<\inf_{\nu}\gamma^{\nu}\le\sup_{\nu}\gamma^{\nu}\le\gamma^{\max}\le1\quad\mbox{and}\quad2c_{\tilde{U}}>\gamma^{\max}L_{\nabla},\label{eq:constant_step_size}
\end{equation}
 then $\{\bx^{\nu}\}$ is bounded and
each of its limit points is a stationary point of problem \ref{eq:nonnonU}.

\noindent (c) If the step-size $\gamma^{\nu}$ is chosen so that
\begin{equation}
\gamma^{\nu}\in(0,1],\quad\gamma^{\nu}\to0,\quad\mbox{and}\quad\sum_{\nu}\gamma^{\nu}=+\infty,\label{eq:diminishing_step_size}\vspace{-0.1cm}
\end{equation}
then $\{\bx^{\nu}\}$ is bounded and
at least one of its limit points is stationary.
If, in addition,   Assumption 4 holds and $\mathcal{X}$ is compact, every limit point of $\{\bx^{\nu}\}$ is  stationary.

Furthermore, if the algorithm does not stop after a finite number of steps, none of the stationary  points above is a local maximum of
$U$.
\end{theorem}

\begin{proof}The proof is quite involved and is given in the appendix; rather classically, its crucial point is showing that
$\underset{_{\nu\rightarrow\infty}}{\text{{lim(inf)}}}\,\|\hat{\bold x}({\bold x}^{\nu})-{\bold x}^{\nu}\|=0.$
\end{proof}\vspace{-0.2cm}
%
%
%


\subsection{Discussions on Algorithm \ref{algoC}}

Algorithm \ref{algoC} describes  a novel family of inner
convex approximation  methods for problem \ref{eq:nonnonU}. Roughly speaking, it
consists in solving the sequence of strongly convex problems  \ref{eq:k2}
wherein the original objective function $U$ is replaced by the strongly
convex (simple) approximation $\tilde{U}$, and the nonconvex constraints
$g_{j}$s with the convex upper estimates $\tilde{g}_{j}$s; convex
constraints, if any, are kept unaltered. A step-size in the update of
the iterates $\bx^{\nu}$ is also  used, in the form of a convex
combination via $\gamma^{\nu}\in(0,1]$ (cf. Step 3). Note that the
iterates $\{\bx^{\nu}\}$ generated by the algorithm are all feasible for
the original problem \ref{eq:nonnonU}. Convergence
is guaranteed under mild assumptions that offer a lot of flexibility
in the choice of the approximation functions and free parameters {[}cf.
Theorem \ref{th:conver}(b) and (c){]}, making the proposed scheme appealing
for many applications. We provide next some examples of candidate
approximants, covering a variety of situations and problems
of practical interest.

\subsubsection{\noindent On the approximations $\tilde{g}_{j}$s}

\noindent As already mentioned, while assumption C3 might look rather
elusive, in many practical cases an upper approximate function for
the nonconvex constraints $g_{j}$s is close at hand. Some examples
of $\tilde{g}_{j}$ satisfying Assumption 3 (and in particular C3)
are given next; specific applications where such approximations are used
 are discussed in detail in Part II of the paper \cite{FacchineiLamparielloScutariTSP14-PartII}.

\noindent \emph{Example \#1$-$ Nonconvex constraints with Lipschitz
gradients}. If the nonconvex function $g_{j}$ does not have a special
structure but Lipschitz continuous gradient on $\mathcal{K}$ with
constant $L_{\nabla g_{j}}$, the following convex approximation function
is a global upper bound of $g_{j}$: for all $\bx\in\mathcal{K}$
and $\by\in\mathcal{X}$,\vspace{-0.1cm}
\begin{equation}
\tilde{g}_{j}(\bx;\by)\triangleq g_{j}(\by)+\nabla_{\bx}g_{j}(\by)\trt(\bx-\by)+\frac{L_{\nabla g_{j}}}{2}\|\bx-\by\|^{2}\ge g_{j}(\bx).\label{eq:Lip}
\end{equation}
\emph{Example \#2$-$ Nonconvex constraints with (uniformly) bounded
Hessian matrix}. Suppose that $g_{j}$ is (nonconvex) $\mathcal{C}^{2}$
with second order bounded derivatives on $\mathcal{K}$. Then, one
can find a matrix $\mathbf{G}\succ\mathbf{0}$ such that $\nabla_{\mathbf{x}}^{2}g_{j}(\mathbf{x})+\mathbf{G}\succeq0$
for all $\mathbf{x}\in\mathcal{K}$. For instance, one can set $\mathbf{G}=\left|\min_{\mathbf{x}\in\mathcal{K}}\lambda_{\min}(\nabla_{\mathbf{x}}^{2}g_{j}(\mathbf{x}))\right|\cdot\mathbf{I}$,
with $\lambda_{\min}(\nabla_{\mathbf{x}}^{2}g_{j}(\mathbf{x}))$ denoting
the minimum eigenvalue of $\nabla_{\mathbf{x}}^{2}g_{j}(\mathbf{x})$
(which is a negative quantity if $g_{j}$ is nonconvex). Then, the
unstructured nonconvex constraint $g_{j}$ can be equivalently written
as a DC function:
\begin{equation}
g_{j}(\bx)=\underset{\triangleq g_{j}^{+}(\bx)}{\underbrace{g_{j}(\bx)+\dfrac{{1}}{2}\,\mathbf{x}^{T}\mathbf{G}\mathbf{x}}}-\underset{\triangleq g_{j}^{-}(\bx)}{\underbrace{\dfrac{{1}}{2}\,\mathbf{x}^{T}\mathbf{G}\mathbf{x}},}\vspace{-0.2cm}\label{eq:DC_constraint}
\end{equation}
where $g_{j}^{+}$ and $g_{j}^{-}$\emph{ }are\emph{ }two \emph{convex
}continuously differentiable functions. An approximant $\tilde{g}_{j}$
of $g_{j}$  satisfying Assumption
3 can then  readily be obtained by linearizing $g_{j}^{-}(\bx)$; see Example
$\#3$ below for details.

The two examples above cover successfully quite general \emph{unstructured
}functions $g_{j}$. However, in some cases, the function parameters
involved in the approximations$-$the constants $L_{\nabla g_{j}}$
or $\left|\min_{\mathbf{x}\in\mathcal{K}}\lambda_{\min}(\nabla_{\mathbf{x}}^{2}g_{j}(\mathbf{x}))\right|-$are
not known exactly but need to be estimated; if the estimates are not
tight, the resulting $\tilde{g}_{j}$ might be a loose overestimation
of $g_{j}$, which may negatively affect the practical convergence
of Algorithm 1. Other approximations can be obtained when $g_{j}$
has further structure to exploit, as discussed in the next examples.

\noindent \emph{Example \#3$-$ Nonconvex constraints with DC structure}.
Suppose that  $g_{j}$ has a DC structure, that is,\vspace{-0.2cm}
\[
g_{j}(\bx)=g_{j}^{+}(\bx)-g_{j}^{-}(\bx)
\]
is the difference of two convex and continuously differentiable functions
$g_{j}^{+}$ and $g_{j}^{-}$. By linearizing the concave part $-g_{j}^{-}$
and keeping the convex part $g_{j}^{+}$ unchanged, we obtain the
following convex upper approximation of $g_{j}$: for all $\bx\in\mathcal{K}$
and $\by\in\mathcal{X}$,\vspace{-0.1cm}
\begin{equation}
\tilde{g}_{j}(\bx;\by)\triangleq g_{j}^{+}(\bx)-g_{j}^{-}(\by)-\nabla_{\bx}g_{j}^{-}(\by)\trt(\bx-\by)\ge g_{j}(\bx).\label{eq:DC}
\end{equation}

\noindent \emph{Example \#4$-$ Bi-linear constraints. }Suppose that\emph{
}$g_{j}$ has a bi-linear structure, that is,\vspace{-0.1cm}
\begin{equation}
g_{j}(x_{1},x_{2})=x_{1}\cdot x_{2}.\label{eq:bi-convex_constraints}
\end{equation}
%
%
Observe preliminarily that $g_{j}(x_{1},x_{2})$ can be rewritten as a DC function:\vspace{-0.1cm}
\begin{equation}
g_{j}(x_{1},x_{2})=\dfrac{{1}}{2}\,(x_{1}+x_{2})^{2}-\dfrac{{1}}{2}\,(x_{1}^{2}+x_{2}^{2}).\label{eq:bi-convex_constraints_app2}
\end{equation}

\noindent A valid $\tilde{g}_{j}$ can be then obtained linearizing
the concave part in (\ref{eq:bi-convex_constraints_app2}): for any
given $(y_{1},y_{2})\in\mathbb{{R}}^{2}$,
\begin{eqnarray*}
\tilde{g}_{j}\left(x_{1},x_{2};\, y_{1},y_{2}\right) & \triangleq & \dfrac{{1}}{2}\,(x_{1}+x_{2})^{2}-\dfrac{{1}}{2}\,(y_{1}^{2}+y_{2}^{2})\\
 &  & -y_{1}\cdot(x_{1}-y_{1})-y_{2}\cdot(x_{2}-y_{2}).
\end{eqnarray*}
In Part II of the paper \cite{FacchineiLamparielloScutariTSP14-PartII}
we show that the constraint functions of many resource allocation problems
in wireless systems and networking fit naturally in Examples 1-4 above.

\subsubsection{\noindent On the approximation $\tilde{U}$\label{sub:On-the-approximation_U}}

\noindent The function $\tilde{{U}}$ should be regarded as a (possibly
simple) convex approximation that preserves the first order
properties of $U$. Some instances of valid $\tilde{U}$s for a specific
$U$ occurring in practical applications are discussed next.

\noindent \emph{Example \#5$-$ Block-wise convex $U(\bx_{1},\ldots,\bx_{n})$}.
In many applications, the vector of variables $\bx$ is partitioned
in blocks $\bx=(\bx_{i})_{i=1}^{I}$ and the function $U$ is convex
in each block $\bx_{i}$ separately, but not jointly. A natural approximation
for such a $U$ exploring its ``partial\textquotedbl{} convexity
is\vspace{-0.1cm}
\begin{equation}
\tilde{U}(\bx;\by)=\sum_{i=1}^{I}\tilde{U}_{i}(\bx_{i};\by),\label{eq:separable_U}
\end{equation}
with each $\tilde{U}_{i}(\bx_{i};\by)$  defined as
\begin{equation}
\tilde{U}_{i}(\bx_{i};\by)\triangleq U(\bx_{i},\by_{-i})+\dfrac{\tau_{i}}{2}(\bx_{i}-\by_{i})^{T}\mathbf{H}_{i}(\by)(\bx_{i}-\by_{i}),\label{U_i_tilde}
\end{equation}
where $\by\triangleq(\by_{i})_{i=1}^{I}$, $\by_{-i}\triangleq(\by_{j})_{j\neq i}$,
and $\mathbf{H}_{i}(\by)$ is any uniformly positive definite matrix
(possibly depending on $\by$). Note that the quadratic term in (\ref{U_i_tilde})
can be set to zero if $U(\bx_{i},\by_{-i})$ is strongly convex in
$\bx_{i}$, uniformly for all feasible $\by_{-i}$.  An alternative choice for $\tilde{U}_{i}(\bx_{i};\by)$
is
\[
\begin{array}{l}
\tilde{U}_{i}(\bx_{i};\by)\triangleq\nabla_{\mathbf{x}_{i}}U(\mathbf{y})^{T}(\mathbf{x}_{i}-\mathbf{y}_{i})\\
\qquad\quad+\dfrac{{1}}{2}\,(\mathbf{x}_{i}-\mathbf{y}_{i})^{T}\nabla_{\mathbf{x}_{i}}^{2}U(\mathbf{y})(\mathbf{x}_{i}-\mathbf{y}_{i})+\dfrac{{\tau_{i}}}{2}\left\Vert \mathbf{x}_{i}-\mathbf{y}_{i}\right\Vert ^{2},
\end{array}
\]
where $\nabla_{\mathbf{x}_{i}}^{2}U(\mathbf{y})$ is the Hessian of
$U$ w.r.t. $\bx_{i}$ evaluated in $\by$. One can also use any positive
definite ``approximation'' of $\nabla_{\mathbf{x}_{i}}^{2}U(\mathbf{y})$.
Needless to say, if $U(\bx_{1},\ldots,\bx_{n})$ is \emph{jointly}
convex in all the variables' blocks, then $\tilde{U}(\bx;\by)$ can
be chosen so that
\begin{equation}\label{eq:convex plus}
\tilde{U}(\bx;\by)\triangleq U(\bx)+\sum_{i}\frac{\tau_{i}}{2}\|\bx_{i}-\by_{i}\|^{2},
\end{equation}
where $\frac{\tau_{i}}{2}\|\bx_{i}-\by_{i}\|^{2}$ is not needed if
$U(\bx_{i},$ $\bx_{-i})$ is strongly convex in $\bx_{i}$, uniformly
for all feasible $\bx_{-i}$.\textcolor{red}{\smallskip{}
}

\noindent \emph{Example \#6$-$(Proximal) gradient-like approximations}.
If no convexity whatsoever is present in $U$, mimicking proximal-gradient
methods, a valid choice of $\tilde{U}$ is the first order approximation
of $U$, that is, $\tilde{U}(\bx;\by)=\sum_{i=1}^{I}\tilde{U}_{i}(\bx_{i};\by)$,
with each
\[
\tilde{U}_{i}(\bx_{i};\by)\triangleq\nabla_{\mathbf{x}_{i}}U(\mathbf{y})^{T}(\mathbf{x}_{i}-\mathbf{y}_{i})+\frac{\tau_{i}}{2}\left\Vert \mathbf{x}_{i}-\mathbf{y}_{i}\right\Vert ^{2}.
\]
 Note that even though classical (proximal) gradient descent methods
(see, e.g., \cite{Bertsekas_Book-Parallel-Comp}) share the same approximation
function, they are not applicable to problem \ref{eq:nonnonU}, due to the nonconvexity
of the feasible set.\textcolor{red}{\smallskip{}
}

\noindent \emph{Example \#7$-$ Sum-utility function}. In multi-agent
scenarios, the objective function $U$ is generally written as $U(\bx)\triangleq\sum_{i=1}^{I}f_{i}(\bx_{1},\ldots,$
$\bx_{I})$, that is, the sum of the utilities $f_{i}(\bx_{1},\ldots,\bx_{I})$
of $I$ agents, each  controlling the variables $\bx_{i}$.
A typical situation is when the $f_{i}$s are convex in some agents' variables.
To capture this property, let us define by
\[
\mathcal{S}_{i}\triangleq\left\{ j\,:\, f_{j}(\bullet,\bx_{-i})\,\text{is convex in }\bx_{i},\,\forall(\bx_{i},\bx_{-i})\in\mathcal{K}\right\}
\]
 the set of indices of all the functions $f_{j}(\bx_{i},\bx_{-i})$
that are convex in $\bx_{i}$, for any feasible $\bx_{-i}$, and let
${\mathcal{C}}_{i}\subseteq{\mathcal{S}}_{i}$ be any subset of ${\mathcal{S}}_{i}$.
Then, the following approximation function $\tilde{U}(\bx;\by)$ satisfies
Assumption 2 while exploiting the partial convexity of $U$ (if any):
$\tilde{U}(\bx;\by)=\sum_{i=1}^{I}\tilde{U}_{\mathcal{C}_{i}}(\bx_{i};\by)$,
with each $\tilde{U}_{\mathcal{C}_{i}}$ defined as
\[
\begin{array}{l}
\tilde{U}_{\mathcal{C}_{i}}(\bx_{i};\by)\triangleq{\displaystyle {\sum_{j\in{\mathcal{C}}_{i}}}}f_{j}(\bx_{i},\by_{-i})\medskip
+{\displaystyle {\sum_{k\notin{\mathcal{C}}_{i}}}}\nabla_{\bx_{i}}f_{k}(\by)\trt(\bx_{i}-\by_{i})\smallskip\\
\qquad\quad\quad\qquad+\dfrac{{\tau_{i}}}{2}\,(\mathbf{x}_{i}-\mathbf{y}_{i})^{T}\mathbf{H}_{i}(\mathbf{y})(\mathbf{x}_{i}-\mathbf{y}_{i}),
\end{array}
\]
where $\mathbf{H}_{i}(\by)$ is any uniformly positive definite matrix.
Roughly speaking, for each agent $i$ we built an approximation function such that the convex part of $U$ w.r.t. $\bx_i$ may be preserved while  the nonconvex part is linearized.

\noindent \emph{Example \#8$-$ Product of functions}.
The function $U$ is often the product of functions (see Part II \cite{FacchineiLamparielloScutariTSP14-PartII} for some examples); we consider here the product of two functions, but the proposed approach can be readily extended to the case of three or more functions or to the sum of such product terms.
Suppose that  $U(\bx) = f_1(\bx) f_2(\bx)$, with $f_1$ and $f_2$ convex and positive.
 In view of the expression of the gradient of $U$,
$\nabla_{\bx} U = f_2 \nabla_{\bx} f_1 + f_1 \nabla_{\bx} f_2$, it seems natural to consider the approximation
\[
\tilde U(\bx;\by) = f_1(\bx) f_2(\by)  + f_1(\by) f_2(\bx) +
\frac{\tau_{i}}{2}(\mathbf{x}-\mathbf{y})^{T}\mathbf{H}(\mathbf{y})(\mathbf{x}-\mathbf{y}),
\]
where, as usual, $\mathbf{H}(\by)$ is a uniformly positive definite matrix; this term can be omitted if $f_1$ and $f_2$ are bounded away from zero on the feasible set
and $f_1+ f_2$ is strongly convex  (for example if one of the two functions is strongly convex).
It is clear that this $\tilde U$ satisfies Assumption 2. In case $f_1$ and $f_2$ are still positive but not necessarily convex, we can
use the expression
\[
\tilde U(\bx;\by) = \tilde f_1(\bx;\by) f_2(\by)  + f_1(\by) \tilde f_2(\bx;\by),
\]
where $\tilde f_1$ and $\tilde f_2$ are any legitimate approximations for $f_1$ and $f_2$, for example those
considered in Examples 5-7 above. Finally, if $f_1$ and $f_2$ can take non-positive values, we can write
\[
\tilde U(\bx;\by) = \tilde h_1(\bx, \by) +  \tilde h_2(\bx,\by),
\]
where $h_1(\bx, \by) \triangleq \tilde f_1(\bx;\by) f_2(\by)$, $h_2(\bx, \by) \triangleq f_1(\by)\tilde f_2(\bx;\by)$, and $\tilde h_1$
and $\tilde h_2$ are legitimate approximations for $h_1$ and $h_2$, for example, again, those
considered in Examples 5-7.
Note that in the last cases we no longer need the quadratic term because it is already included in the approximations $\tilde f_1$ and $\tilde f_2$, and $\tilde h_1$ and $\tilde h_2$, respectively. As a final remark, it is important to point out that the the $U$s discussed above belong to a class of nonconvex functions for which it does not seem possible to find a global convex upper bound; therefore, all current SCA techniques (see, e.g.,  \cite{MarWrig78,RazaviyaynHongLuo_SIAMOPT13,QuDie11}) are not applicable. \smallskip

We conclude the discussion on the approximation functions observing
that, in  Examples 5-7, the proposed $\tilde{U}(\bx;\by)$s are all separable in the blocks $\bx_{i}$ for any given $\mathbf{y}$; in Example 8, instead, the separability is problem dependent and should be examined on a case-by-case basis.
The separability of the $\tilde U$s  paves the way to parallelizable and distributed versions
of Algorithm 1;  we discuss this topic more in detail in Sec. \ref{sec:Distributed-implementation}.\smallskip{}

\subsubsection{\noindent On the choice of the step-size rule $\gamma^{\nu}$}\label{subsec:gamma}

\noindent Theorem \ref{th:conver} states that Algorithm \ref{algoC}
converges either employing a constant step-size rule {[}cf. (\ref{eq:constant_step_size}){]}
or a diminishing step-size rule [cf. (\ref{eq:diminishing_step_size})].

If a constant step-size is used, one can set in (\ref{eq:constant_step_size})
$\gamma^{\nu}=\gamma\le\gamma^{\max}$ for every $\nu$, and choose
any $\gamma^{\max}\in(0,1]$ and $c_{\tilde{U}}$ so that  $2c_{\tilde{U}}>\gamma^{\max}L_{\nabla U}$
(recall that $c_{\tilde{U}}$ is the constant of strong convexity
of the approximation $\tilde{{U}}$ and, thus, is a degree of freedom).
This can be done in several ways. For instance, if the chosen $\tilde{{U}}$
contains a proximal term with gain $\tau>0$, i.e., a term of the type $(\tau/2) \|\bx - \by\|^2$, then the inequality
$2c_{\tilde{U}}>\gamma^{\max}L_{\nabla U}$ is readily satisfied setting
$2\tau/\gamma^{\max}>\, L_{\nabla U}$ (we used $c_{\tilde{U}}\geq\tau$).
Note that this simple (but conservative) condition imposes a constraint
only on the ratio $\tau/\gamma^{\max}$, leaving free the choice of
one of the two parameters. An interesting special case worth mentioning
is when $\gamma^{\max}=1$ and $2\tau>L_{\nabla U}$: the choice  $\gamma^\nu=1$  leads to an instance of Algorithm 1 with no memory, i.e., $\mathbf{x}^{\nu+1}=\hat{\bold x}({\bold x}^{\nu})$, for all $\nu$.

When the Lipschitz constant $L_{\nabla U}$ cannot be estimated, one
can use a diminishing step-size rule, satisfying the standard conditions
(\ref{eq:diminishing_step_size}). A rule that we found to work
well in practice is, see \cite{ScutFacchSonPalPang13}:\vspace{-0.3cm}
\begin{equation}
\gamma^{\nu}=\gamma^{\nu-1}(1-\varepsilon\gamma^{\nu-1}),\quad\nu\geq1,\label{eq:step-size}
\end{equation}
with  $\gamma^{0}\in(0,1]$ and $\varepsilon\in(0,1)$. Other effective
rules can be found in  \cite{ScutFacchSonPalPang13}. Notice that,
while this rule may still require some tuning for optimal behavior,
it is quite reliable, since in general we are not using a (sub)gradient
direction, so that many of the well-known practical drawbacks associated
with a (sub)gradient method with diminishing step-size are mitigated
in our setting. Furthermore, this choice of step-size does not require
any form of centralized coordination and, thus, provides a favorable feature
in distributed environments.

We remark that it is possible to prove the convergence of Algorithm 1
also using other step-size rules, such as a standard Armijo-like line-search
procedure. We omit the discussion of line-search based approaches because such options
are not in line with our goal of developing distributed algorithms,
see Sec. \ref{sec:Distributed-implementation}.
In \cite{BeckBenTTetr10} it is shown that, in the specific case of a {\em strongly convex} $U$ and, in our terminology, $\tilde U =U$ and ${\cal K} = \Re^n$, by choosing
$\gamma^\nu =1$ at every iteration, one can prove the stationarity of every limit point of the sequence generated by Algorithm 1
(assuming regularity). We can easily derive this particular result from our general analysis, see
Remark \ref{rem:btone} in the Appendix. Here we only mention that, attractive as this result may be, the strong convexity of $U$ is a very restrictive assumption, and forcing $\tilde U =U$ does not permit the development of distributed versions of Algorithm 1.

Finally, as a technical note, it is interesting to contrast the different
kinds of convergence
that one can obtain by choosing
a constant or a diminishing step-size rule. In the former case, every
 limit point of the sequence generated by Algorithm 1 is guaranteed to be a stationary solution of the
original nonconvex problem \ref{eq:nonnonU}, whereas, in the latter
case, there exists at least a limit point being stationary, which  thus is a weaker condition. On the other hand, a diminishing step-size has
been observed to be numerically more efficient than a constant one.
In order to obtain, also with a diminishing step-size rule, the strong convergence behavior that can be guaranteed
when a constant step-size is employed, one needs extra conditions
on the approximation functions $\tilde{{U}}$ and $\tilde{{g}}$
(cf. Assumptions 4); these conditions are quite standard and easy
to be satisfied in most practical applications (as those studied
in Part II \cite{FacchineiLamparielloScutariTSP14-PartII}).\vspace{-0.2cm}

\subsection{\noindent Related works\label{sub:Related-works}}

Our approach draws on the  
SCA paradigm, which has been widely explored in the literature,
see \cite{MarWrig78,BeckBenTTetr10,RazaviyaynHongLuo_SIAMOPT13,ScutFacchSonPalPang13,
AlvScutPang13,QuDie11}. 
However, our framework and convergence conditions unify and extend
current SCA methods in several directions, as outlined next.

\noindent $-$\emph{On the approximation functions}: Conditions on
the approximation function $\tilde{U}$ as in Assumption 2 are relatively
weak: this feature allows us to enlarge significantly the class of utility functions $U$s
that can be successfully handled by Algorithm 1. A key difference
with current SCA methods {[}applicable to special cases of \ref{eq:nonnonU}{]}
\cite{MarWrig78,RazaviyaynHongLuo_SIAMOPT13} and DC programs \cite{QuDie11}
is that the
 approximation $\tilde{U}(\bx;\by)$ need not
be a tight \emph{global upper bound} of the objective function $U(\bx)$
for every $\bx \in \mathcal{K}$ and $\by\in\mathcal{X}$ {[}cf. Assumption 2{]}. This fact represents
a big step forward in the literature of SCA methods; Part II of the paper
\cite{FacchineiLamparielloScutariTSP14-PartII} provides a solid evidence
of the wide range of applicability of the proposed framework.

\noindent \emph{$-$Convergence conditions: }There are only a few
SCA-based methods in the literature handling nonconvex constraints,
namely \cite{MarWrig78,BeckBenTTetr10,QuDie11}, and the existing
convergence results are quite weak. In particular, \cite[Th. 1]{MarWrig78}
states that if {\em the whole sequence converges,} then the algorithm
converges to a stationary point; however, in general, it is hardly possible
to show that the sequence generated by the algorithms does converge.
In \cite{BeckBenTTetr10}, (subsequence)
convergence to regular points is proved, but only for nonconvex problems
with \emph{strongly convex} objective functions; this fact restricts considerably
the range of applicability of this result (for instance, none
of the problems that we study in Part II \cite{FacchineiLamparielloScutariTSP14-PartII}
have strongly convex objective functions). Finally, \cite{QuDie11}
can handle only (possibly nonsmooth) nonconvex problems  whose objective
functions and constraints have a DC form. To the best of our knowledge,
this work is the first attempt towards the generalization of SCA methods
to nonconvex problems having general nonconvex objective functions
and constraints.

\noindent \emph{$-$Distributed implementation: }A second key and
unique feature of Algorithm 1,  missing in current SCA schemes
\cite{MarWrig78,BeckBenTTetr10,QuDie11}, is that it easily leads
to distributed implementations, as we will discuss in Sec.
\ref{sec:Distributed-implementation}. This feature, along with the
feasibility of the iterates, represents a key difference also with classical
techniques \cite{LawTit01,FerrMan94,SagSol02,Sol98} that have been proposed in the
literature to deal with nonconvex optimization problems.\vspace{-0.1cm}

\section{Distributed implementation\label{sec:Distributed-implementation}}

In many applications, e.g., multi-agent optimization or distributed
networking, it is desirable to keep users coordination and communication
overhead at minimum level. In this section we discuss distributed versions of
Algorithm 1. Of course,  we need to assume that problem \ref{eq:nonnonU} has some
suitable structure, and that consistent choices on $\tilde U$ and $\tilde g$  are made.
Therefore, in this section we consider the following additional assumptions.\smallskip{}

\noindent \textbf{Assumption 6 (Decomposabilty)}. Given \ref{eq:nonnonU}, suppose
that:

\noindent D1) the set $\mathcal{K}$ has a Cartesian structure, i.e.,
$\mathcal{K}=\mathcal{K}_{1}\times\cdots\times\mathcal{K}_{I}$, with
each $\mathcal{K}_{i}\subset\mathbb{R}^{n_{i}}$, and $\sum_{i}n_{i}=n$;
 $\bx\triangleq(\bx_{i})_{i=1}^{I}$ is partitioned accordingly,
with each $\bx_{i}\in\mathcal{K}_{i}$;

\noindent D2) the approximate function $\tilde{U}(\bx;\by)$ satisfying
Assumption 2 is chosen so that $\tilde{U}(\bx;\by)=\sum_{i}\tilde{U}_{i}(\bx_{i};\by)$;

\noindent D3) each approximate function $\tilde{g}_{j}(\bx;\by)$
satisfying Assumption 3 is (block) separable in the $\bx$-variables,
for any given $\by$, that is, $\tilde{g}_{j}(\bx;\by)=\sum_{i}\tilde{g}_{j}^{i}(\bx_{i};\by)$,
for some $\tilde{g}_{j}^{i}:\mathcal{K}_{i}\times\mathcal{X}\rightarrow\mathbb{R}$.\medskip

\noindent
Condition D1 is a very natural assumption on problem \ref{eq:nonnonU} and is usually satisfied  when
a distributed implementation is called for.  If problem \ref{eq:nonnonU} does not satisfy this assumption, it is
not realistic to expect that efficient distributed solution methods can be devised; D2 and D3, instead, are assumptions on our algorithmic choices. In particular, condition D2  permits  many choices for $\tilde{U}$.
For instance, as already discussed at the end of the subsection  "On the approximation $\tilde U$",
essentially all $\tilde{U}$s introduced in Examples 5-7 (and possibly some of the $\tilde U$s in Example 8) satisfy D2. The critical
condition
in Assumption 6 is D3. Some examples of constraints functions $g_{j}$
for which one can find a $\tilde{g}_{j}(\bx;\by)$ satisfying D3 are:

\noindent $-$\emph{Individual nonconvex constraints}: Each $g_{j}$
(still nonconvex) depends only on one of the block variables $\bx_{1},\ldots,\bx_{I}$,
i.e, $g_{j}(\bx)=g_{j}^{i}(\bx_{i})$, for some $g_{j}^{i}:\mathcal{K}_{i}\rightarrow\mathbb{R}$
and $i$;

\noindent $-$\emph{Separable nonconvex constraints}: Each $g_{j}$
has the form $g_{j}(\bx)=\sum_{i}g_{j}^{i}(\bx_{i})$, with $g_{j}^{i}:\mathcal{K}_{i}\rightarrow\mathbb{R}$;

\noindent $-$\emph{Nonconvex constraints with Lipschitz gradients}:
Each $g_{j}$ is not necessarily separable but has Lipschitz gradient
on $\mathcal{K}$. In this case one can choose, e.g., the approximation
$\tilde{g}_{j}$ as in (\ref{eq:Lip}).\smallskip{}

It is important to remark that, even for problems \ref{eq:nonnonU} {[}or  \ref{eq:k2}{]}
for which it looks hard to satisfy D3, the introduction of proper slack variables
can help to decouple the constraint functions, making thus possible
to find a $\tilde{g}_{j}$ that satisfies D3; we refer the reader to Part II of the paper
\cite{FacchineiLamparielloScutariTSP14-PartII} for some non trivial examples where this technique is applied.

For notational simplicity, let us introduce the vector function $\tilde{\mathbf{g}}^{i}(\bx_{i};\bx^{\nu})\triangleq(\tilde{g}_{j}^{i}(\bx_{i};\bx^{\nu}))_{j=1}^{m}$,
for $i=1,\ldots,I$. Under Assumption 6, each subproblem  \ref{eq:k2}
becomes

\begin{equation}
\begin{array}{ll}
\underset{\bx}{\textnormal{min}} & \sum_{i=1}^{I}\tilde{U}_{i}(\bx_{i};\bx^{\nu})\smallskip\\
\mbox{ s.t.} & \hspace{-0.6cm}\left.\begin{array}{ll}
 & \tilde{\mathbf{g}}(\bx;\bx^{\nu})\triangleq\sum_{i}\tilde{\mathbf{g}}^{i}(\bx_{i};\bx^{\nu})\le\mathbf{0},\\[5pt]
 & \bx_{i}\in\mathcal{K}_{i},\quad i=1,\ldots,I.
\end{array}\right.
\end{array}\tag{${\tilde{\mathcal{P}}_{\bx^\nu}}$}\label{eq:k2_dec}
\end{equation}

With a slight abuse of notation, we will still denote the feasible
set of \ref{eq:k2_dec} by $\mathcal{X}(\mathbf{x}^{\nu})$.

The block separable structure of both the objective function and the constraints lends itself to a parallel decomposition of
the subproblems \ref{eq:k2_dec} in the primal or dual domain: hence, it allows the distributed implementation of Step 2 of Algorithm 1. In the next
section we briefly show how to customize standard primal/dual decomposition
techniques (see, e.g., \cite{Palomar-Chiang_ACTran07-Num,Bertsekas_Book-Parallel-Comp})
in order to solve subproblem \ref{eq:k2_dec}. We conclude this section observing
that, if there are only individual constraints in \ref{eq:nonnonU}, given $\bx^{\nu}$,
each \ref{eq:k2_dec} can be split in $I$ independent subproblems
in the variables $\bx_{i}$, even if the original nonconvex $U$ is
\emph{not separable}. To the best of our knowledge, this is the first
attempt to obtain distributed algorithms for a nonconvex problem in
the general form \ref{eq:nonnonU}.\vspace{-0.2cm}

\subsection{Dual decomposition methods\label{sub:dual-decomposition}}

Subproblem \ref{eq:k2_dec} is convex and can be solved in a distributed
way if the constraints $\tilde{\mathbf{g}}(\bx;\bx^{\nu})\le\mathbf{0}$
are dualized.  The dual
problem associated with each \ref{eq:k2_dec} is: given $\mathbf{x}^{\nu}\in\mathcal{X}(\mathbf{x}^{\nu})$, 
\begin{equation}
\begin{aligned}\underset{\boldsymbol{\lambda}\geq\mathbf{0}}{\text{max}}\,\, d(\boldsymbol{\lambda};\mathbf{x}^{\nu})\end{aligned} \label{eq:dualpro}
\end{equation}\vspace{-0.2cm}
with
\begin{equation}
\begin{array}{c}
d(\boldsymbol{\lambda};\mathbf{x}^{\nu})\triangleq\underset{\mathbf{x}\in\mathcal{K}}{\text{min}}\,\left\{ \sum_{i=1}^{I}\left(\tilde{U}_{i}(\bx_{i};\bx^{\nu})+\boldsymbol{\lambda}^{T}\mathbf{\tilde{\mathbf{g}}}^{i}(\bx_{i};\bx^{\nu})\right)\right\} \end{array}\label{eq:dual function}
\end{equation}
Note that, for $\boldsymbol{\lambda} \ge 0$, by Assumptions 2 and 3, the minimization in (\ref{eq:dual function})
has a unique solution, which will be denoted by $\widehat{{\mathbf{x}}}(\boldsymbol{\lambda};\mathbf{x}^{\nu})\triangleq(\widehat{{\mathbf{x}}}_{i}(\boldsymbol{\lambda};\mathbf{x}^{\nu}))_{i=1}^{I}$,
with
\begin{equation}
\widehat{{\mathbf{x}}}_{i}(\boldsymbol{\lambda};\mathbf{x}^{\nu})\triangleq\underset{\mathbf{x}_{i}\in\mathcal{K}_{i}}{\text{{arg}}\text{min}}\,\left\{ \tilde{U}_{i}(\bx_{i};\bx^{\nu})+\boldsymbol{\lambda}^{T}\mathbf{\tilde{\mathbf{g}}}^{i}(\bx_{i};\bx^{\nu})\right\} .\label{eq:best-response_lagrangian}
\end{equation}
Before proceeding, let us mention the following standard condition.\smallskip{}


\noindent D4) $\tilde{\mathbf{g}}(\bullet;\bx^{\nu})$ is uniformly
Lipschitz continuous on $\mathcal{K}$ with constant $L_{\tilde{g}}$;\smallskip

\noindent We  remind that D4 is implied by condition C7 of Assumption 4; therefore we do not need to assume it if Assumption 4
is invoked. But in order to connect our results below to classical ones, it is good to highlight it as a separate assumption.

The next lemma summarizes some desirable properties of the dual function
$d(\boldsymbol{\lambda};\mathbf{x}^{\nu})$, which are instrumental
to prove the convergence of dual schemes.

\begin{lemma} \label{prop:suf cond inner conv} Given \ref{eq:k2_dec},
under Assumptions 1-3, 5, and 6, the following hold.

\noindent (a) $d(\boldsymbol{\lambda};\mathbf{x}^{\nu})$ is differentiable with respect to $\boldsymbol{\lambda}$
on $\mathbb{R}_{+}^{m}$, with gradient\vspace{-0.2cm}
\begin{equation}
\nabla_{\boldsymbol{{\lambda}}} d(\boldsymbol{\lambda};\mathbf{x}^{\nu})=
\sum_{i}\tilde{\mathbf{g}}^{i}\left(\widehat{{\mathbf{x}}}_{i}(\boldsymbol{\lambda};\mathbf{x}^{\nu});\mathbf{x}^{\nu}\right).
\label{eq:gradient_dual}
\end{equation}

\noindent (b) If in addition D4 holds, then $\nabla_{\boldsymbol{{\lambda}}}d(\boldsymbol{\lambda};\mathbf{x}^{\nu})$
is Lipschitz continuous on $\mathbb{R}_{+}^{m}$ with constant \textcolor{red}{${\color{black}L_{\nabla d}\triangleq L_{\tilde{g}}^{2}\,\sqrt{{m}}/c_{\tilde{U}}}$}.
\end{lemma}
\begin{IEEEproof}
\noindent See Appendix C.
\end{IEEEproof}
The dual-problem can be solved, e.g., using well-known gradient algorithms
\cite{BertsekasNedicOzdaglar_book_convex03}; an instance is given
in Algorithm \ref{alg:Dual-Decomposition}, whose convergence is stated
in Theorem \ref{thm:dual_convergence}. The proof of the theorem follows
from Lemma \ref{prop:suf cond inner conv} and standard convergence
results for gradient projection algorithms (e.g., see \cite[Th. 3.2]{Su-Xu10}
and \cite[Prop. 8.2.6]{BertsekasNedicOzdaglar_book_convex03} for
the theorem statement under assumptions (a) and (b), respectively).
We remark that an assumption made in the aforementioned references is that
subproblem \ref{eq:k2_dec}has a zero-duality
gap and the dual problem \eqref{eq:dualpro}
has a non-empty solution set. In our setting, this is guaranteed by Assumption 5, that ensures that $\mathcal{X}(\mathbf{x}^{\nu})$ satisfies Slater's CQ  (see the discussion around Assumption 5).

In (\ref{eq:sum-rate-dual_update-1}) $[\bullet]_{+}$ denotes the
Euclidean projection onto $\mathbb{R}_{+}$, i.e., $[x]_{+}\triangleq\max(0,x)$.

\begin{algorithm}[h]
\textbf{Data}: $\boldsymbol{\lambda}^{0}\geq\mathbf{0}$, $\mathbf{x}^{\nu}$,
$\{\alpha^{n}\}>0$; set $n=0$.

(\texttt{S.2a}) If $\boldsymbol{\lambda}^{n}$ satisfies a suitable
termination criterion: STOP.

(\texttt{S.2b}) Solve in parallel (\ref{eq:best-response_lagrangian}):
for all $i=1,\dots,I$, compute $\widehat{{\mathbf{x}}}_{i}(\boldsymbol{\lambda}^{n};\mathbf{x}^{\nu})$.

(\texttt{S.2c}) Update $\boldsymbol{\lambda}$ according to
\begin{equation}
\boldsymbol{\lambda}^{n+1}\triangleq\left[\boldsymbol{\lambda}^{n}+\alpha^{n}\,\sum_{i=1}^{I}\tilde{\mathbf{g}}^{i}\left(\widehat{{\mathbf{x}}}_{i}(\boldsymbol{\lambda}^{n};\mathbf{x}^{\nu});\mathbf{x}^{\nu}\right)\right]_{+}.\label{eq:sum-rate-dual_update-1}
\end{equation}

(\texttt{S.2d}) $n\leftarrow n+1$ and go back to (\texttt{S.2a}).

\protect\caption{\hspace{-3pt}\textbf{: \label{alg:Dual-Decomposition}}Dual-based
Distributed Implementation of Step 2 of NOVA Algorithm (Algorithm 1). }
\end{algorithm}

\begin{theorem}\label{thm:dual_convergence} Given \ref{eq:nonnonU},
under Assumptions 1-3, 5, and 6, suppose that one of the two following conditions
is satisfied:

\noindent (a) D4 holds true and $\{\alpha^{n}\}$ is chosen
such that $0<\inf_n\,\alpha^{n}\leq\sup_n\,\alpha^{n}<2/L_{\nabla d}$,
for all $n\geq0$;


\noindent (b)\emph{ }$\nabla_{\boldsymbol{{\lambda}}}d(\boldsymbol{\bullet};\mathbf{x}^{\nu})$
is uniformly bounded on\emph{ }$\mathbb{R}_{+}^{m}$, and $\alpha^{n}$
is chosen such that $\alpha^{n}>0$, $\alpha^{n}\rightarrow0$, $\sum_{n}\alpha^{n}=\infty$,
and $\sum_{n}(\alpha^{n})^{2}<\infty$.

Then, the sequence $\left\{ \boldsymbol{\lambda}^{n}\right\} $
generated by Algorithm \ref{alg:Dual-Decomposition} converges to
a solution of \eqref{eq:dualpro}\emph{, }and the sequence $\{\widehat{{\mathbf{x}}}(\boldsymbol{\lambda}^{n};\mathbf{x}^{\nu})\}$
converges to the unique solution of \ref{eq:k2_dec}. \end{theorem}\smallskip{}

\begin{remark}[On the distributed implementation] {The\,implementation
of Algorithm \ref{algoC} based on Algorithm \ref{alg:Dual-Decomposition}
leads to a double-loop scheme:
given the current value of the multipliers $\boldsymbol{{\lambda}}^{n}$,
the subproblems (\ref{eq:best-response_lagrangian}) can be solved
in parallel across the blocks; once the new values $\widehat{{\mathbf{x}}}_{i}(\boldsymbol{\lambda}^{n};\mathbf{x}^{\nu})$
are available, the multipliers are updated according to (\ref{eq:sum-rate-dual_update-1}).
Note that when $m=1$ (i.e., there is only one shared constraint),
the update in (\ref{eq:sum-rate-dual_update-1}) can be replaced by
a bisection search, which generally converges quite fast. When $m>1$,
the potential slow convergence of gradient updates (\ref{eq:sum-rate-dual_update-1})
can be alleviated using accelerated gradient-based (proximal) schemes;
see, e.g., \cite{Nesterov04,DualAceleration-1}}.

As far as the communication overhead is concerned, the required signaling
 is in the form of message
passing and of course is problem dependent; see Part II of the paper
\cite{FacchineiLamparielloScutariTSP14-PartII} for specific examples.
For instance, in networking applications where there is a cluster-head,
the update of the multipliers can be performed at the cluster, and,
then, it can be broadcasted to the users. In fully decentralized networks instead,
the update of $\boldsymbol{{\lambda}}$ can be done by the users themselves,
by running consensus based algorithms to locally estimate $\sum_{i=1}^{I}\tilde{\mathbf{g}}^{i}\left(\widehat{{\mathbf{x}}}_{i}(\boldsymbol{\lambda}^{n};\mathbf{x}^{\nu});\mathbf{x}^{\nu}\right)$.
This, in general, requires a limited signaling exchange among neighboring
nodes only. Note also that the size of the dual problem (the dimension
of $\boldsymbol{\lambda}$) is equal to $m$ (the number of shared
constraints), which makes Algorithm\emph{ }\ref{alg:Dual-Decomposition}
scalable in the number of blocks (users). \end{remark}\vspace{-0.2cm}

\subsection{Primal decomposition methods\label{sub:primal-decomposition}}

Algorithm \ref{alg:Dual-Decomposition} is based on the relaxation
of the shared constraints into the Lagrangian, resulting, in general,
in a violation of these constraints during the intermediate iterates.
In some applications this fact may prevent the on-line implementation
of the algorithm. In this section we propose a distributed scheme
that does not suffer from this issue: we cope with the shared constraints
using a primal decomposition technique.

Introducing the slack variables $\mathbf{t}\triangleq(\mathbf{t}_{i})_{i=1}^{I}$,
with each $\mathbf{t}_{i}\in\mathbb{R}^{m}$, \ref{eq:k2_dec}can
be rewritten as
\begin{equation}
\begin{array}{cl}
\underset{(\mathbf{x}_{i},\mathbf{t}_{i})_{i=1}^{I}}{\textrm{min}} & \sum_{i=1}^{I}\tilde{U}_{i}(\bx_{i};\bx^{\nu}),\\
\textrm{s.t.} & \mathbf{x}_{i}\in\mathcal{K}_{i},\quad\forall i=1,\ldots,I,\smallskip\\
 & \tilde{\mathbf{g}}^{i}\left(\mathbf{x}_{i};\mathbf{x}^{\nu}\right)\leq\mathbf{t}_{i},\quad\forall i=1,\ldots,I,\smallskip\\
 & \sum_{i=1}^{I}\mathbf{t}_{i}\leq\mathbf{0}.
\end{array}\label{eq:primal_reformulation}
\end{equation}
When $\mathbf{t}=(\boldsymbol{t}_{i})_{i=1}^{I}$ is fixed, (\ref{eq:primal_reformulation})
can be decoupled across the users: for each $i=1,\ldots,I$, solve
\begin{equation}
\begin{array}{cl}
\underset{\mathbf{x}_{i}}{\textrm{min}} & \tilde{U}_{i}(\bx_{i};\bx^{\nu}),\\
\mbox{s.t.} & \mathbf{x}_{i}\in\mathcal{K}_{i},\\
 & \tilde{\mathbf{g}}^{i}\left(\mathbf{x}_{i};\mathbf{x}^{\nu}\right)\overset{\boldsymbol{{\mu}}_{i}(\mathbf{t}_{i};\mathbf{x}^{\nu})}{\leq}\mathbf{t}_{i},
\end{array}\label{eq:primal_single_user}
\end{equation}
where $\boldsymbol{{\mu}}_{i}(\mathbf{t}_{i};\mathbf{x}^{\nu})$ is
the optimal Lagrange multiplier associated with the inequality constraint
$\tilde{\mathbf{g}}^{i}\left(\mathbf{x}_{i};\mathbf{x}^{\nu}\right)\leq\mathbf{t}_{i}$.
\textcolor{black}{Note that the existence of $\boldsymbol{{\mu}}_{i}(\mathbf{t}_{i};\mathbf{x}^{\nu})$
is guaranteed if (\ref{eq:primal_single_user}) has zero-duality gap
}\cite[Prop. 6.5.8]{BertNedOzd03}\textcolor{black}{{} (e.g., when some
CQ hold), but $\boldsymbol{{\mu}}_{i}(\mathbf{t}_{i};\mathbf{x}^{\nu})$
may not be unique.}\textcolor{red}{{} }Let us denote by $\mathbf{x}_{i}^{\star}(\mathbf{t}_{i};\mathbf{x}^{\nu})$
the unique solution of (\ref{eq:primal_single_user}), given $\mathbf{t}=$$(\boldsymbol{t}_{i})_{i=1}^{I}$.
The optimal partition $\mathbf{t}^{\star}\triangleq(\mathbf{t}_{i}^{\star})_{i=1}^{I}$
of the shared constraints can be found solving the so-called \emph{master}
(convex) problem (see, e.g., \cite{Palomar-Chiang_ACTran07-Num}):
\begin{equation}
\begin{array}{cl}
\underset{\boldsymbol{t}}{\textrm{min}} & P(\boldsymbol{t};\mathbf{x}^{\nu})\triangleq\sum_{i=1}^{I}\tilde{U}_{i}(\mathbf{x}_{i}^{\star}(\mathbf{t}_{i};\mathbf{x}^{\nu});\mathbf{x}^{\nu})\\
\textrm{s.t.} & \sum_{i=1}^{I}\mathbf{t}_{i}\leq\mathbf{0}.
\end{array}\label{eq:primal_master}
\end{equation}
Due to the non-uniqueness of $\boldsymbol{{\mu}}_{i}(\mathbf{t}_{i};\mathbf{x}^{\nu})$,
the objective function in (\ref{eq:primal_master}) is nondifferentiable;
problem (\ref{eq:primal_master}) can be solved by subgradient methods.
The partial subgradient of $P(\boldsymbol{t};\mathbf{x}^{\nu})$ with respect to the first argument evaluated at $(\mathbf{t};\bx^\nu)$
is \textcolor{black}{{}
\[
\partial_{\mathbf{t}_{i}}P(\mathbf{t};\mathbf{x}^{\nu})=-\boldsymbol{{\mu}}_{i}(\mathbf{t}_{i};\mathbf{x}^{\nu}),\quad i=1,\ldots,I.
\]
} We refer to \cite[Prop. 8.2.6]{BertsekasNedicOzdaglar_book_convex03}
for standard convergence results for subgradient projection algorithms.

\section{Conclusions\label{sec:Conclusions}}

In this Part I of the two-part paper, we proposed a novel general
algorithmic framework based on convex approximation  techniques for the solution of
nonconvex smooth optimization problems: we point out that the nonconvexity may occur both in the objective function and in the constraints. Some key
novel features of our scheme are: i) it maintains feasibility \emph{and}
leads to \emph{parallel and distributed} solution methods for a very
general class of nonconvex problems; ii) it offers a lot of flexibility
in choosing the approximation functions, enlarging significantly
the class of problems that can be solved with provable convergence;
iii) by choosing different approximation functions, different (distributed)
schemes can be obtained: they are all convergent, but differ for (practical) convergence
speed, complexity, communication overhead, and a priori knowledge
of the system parameters; iv) it includes as special cases several
classical SCA-based algorithms and improves on their convergence properties;
and v) it provides new efficient algorithms also for old problems.
In Part II \cite{FacchineiLamparielloScutariTSP14-PartII} we customize
the developed algorithmic framework to a variety of new (and old)
multi-agent optimization problems in signal processing, communications
and networking, providing a solid evidence of its good performance.
Quite interestingly, even when compared with existing schemes that have been designed
for very specific problems, our algorithms are shown to outperform
them.\vspace{-0.3cm}

\appendix{}

We first introduce some intermediate technical results that are instrumental
to prove Theorem \ref{th:conver}. The proof of Theorem \ref{th:conver}
is given in Appendix \ref{sub:Proof_main_theo}.\vspace{-0.3cm}

\subsection{Intermediate results\label{sub:Technical-preliminaries}}

We first prove Lemma \ref{th:uspropfs-1}-Lemma
\ref{th:basicarg}, providing some key properties of the sequence $\{{\bold x}^{\nu}\}$
generated by Algorithm 1 and of the best-response function $\hat{\bold x}(\bullet)$
defined in (\ref{eq:best-response}). Finally, with Theorem \ref{th:conver-1} we establish some technical conditions
under which a(t least one) regular limit point of the sequence generated
by Algorithm 1 is a stationary solution of the original nonconvex
problem \ref{eq:nonnonU}; the proof of Theorem \ref{th:conver}
will rely on such conditions. We recall that, for the sake of simplicity, throughout
this section we tacitly assume that Assumptions 1-3 and 5 are satisfied. \smallskip

\noindent \noindent \textbf{Lemma \ref{th:uspropfs-1}}. The first
lemma shows, among other things, that Algorithm
1 produces a sequence of points that are feasible for the original problem \ref{eq:nonnonU}.

\begin{lemma}\label{th:uspropfs-1}
The following properties hold.

\noindent(i) ${\by}\in{\cal {\cal X}}(\by)\subseteq{\cal {\cal X}}$
for all ${\by}\in{\cal {\cal X}}$;

\noindent(ii) $\hat{\bold x}({\by})\in{\cal {\cal X}}(\by)\subseteq{\cal {\cal X}}$
for all ${\by}\in{\cal {\cal X}}$.

Moreover, the sequence $\{\mathbf{x}^{\nu}\}$ generated by Algorithm 1 is such that:

\noindent(iii) ${\bold x}^{\nu}\in{\cal {\cal X}}$;

\noindent(iv) ${\bold x}^{\nu+1}\in{\cal {\cal X}}(\bold{x}^{\nu})\cap{\cal {\cal X}}(\bold{x}^{\nu+1})$.
\end{lemma}

\begin{proof} (i) the first implication ${\by}\in{\mathcal {X}}(\by)$
follows from $\tilde{g}_{j}(\by;\by)=g_{j}(\by)\le0$,
for all $j=1,\ldots,m$ {[}due to C2{]}. For the inclusion ${\cal {\cal X}}(\by)\subseteq{\cal {\cal X}}$,
it suffices to recall that, by C3, we have $g_{j}({\bold x})\le\tilde{g}_{j}({\bold x};\by)$
for all ${\bold x}\in{\cal K}$, $\by\in{\mathcal{X}}$,
and $j=1,\ldots,m$, implying that, if ${\bx}\in{\mathcal{X}}(\by)$,
then ${\bx}\in{\cal {\cal X}}$.

\noindent(ii) $\hat{\bold x}(\by)\in{\cal {\cal X}}(\by)$
since it is the optimal solution of $\mathcal P_\by$ (and thus also
feasible).

\noindent(iii) In view of (i) and (ii), it follows by induction and the fact that ${\bold x}^{\nu+1}$ is a convex combination of ${\bold x}^{\nu}\in{\cal {\cal X}}(\bold{x}^{\nu})$ and $\hat{\bold x}({\bold x}^{\nu})\in{\cal {\cal X}}(\bold{x}^{\nu})$, which is a convex subset of $\mathcal{X}$.

\noindent(iv) By (iii), ${\bold x}^{\nu+1}\in{\cal {\cal X}}(\bold{x}^{\nu})$.
Furthermore, we have $\tilde{g}_{j}({\bold x}^{\nu+1};{\bold x}^{\nu+1})=g_{j}({\bold x}^{\nu+1})\le0$,
for all $j=1,\ldots,m$, where the equality follows from C2 and the
inequality is due to ${\bold x}^{\nu+1}\in{\cal {\cal X}}$; thus, ${\bold x}^{\nu+1}\in{\cal {\cal X}}(\bold{x}^{\nu+1})$.
\end{proof}

\medskip{}
\noindent \textbf{Lemma \ref{th:Lemma_descent}}. With Lemma \ref{th:Lemma_descent}, we establish some key properties of the
best-response function $\hat{\bold x}(\bullet)$. We will use the
following definitions. Given ${\bold y},{\bold z}\in{\cal X}$ and
$\rho>0$, let
\begin{equation}
\mathbf{w}_{\rho}(\hat{\bold x}({\bold y}),{\bold z})\triangleq\hat{\bold x}({\bold y})-\rho\nabla_{\bold x}\tilde{U}(\hat{\bold x}({\bold y});{\bold z});\label{eq:w_rho_def}
\end{equation}
and let $P_{{\cal X}({\bold y})}({\bold u})$ denote the Euclidean
projection of ${\bold u}\in\mathbb{R}^{n}$ onto the closed convex
set ${\cal X}({\bold y})$:
\begin{equation}
P_{{\cal X}({\bold y})}({\bold u})=\underset{\mathbf{x}\in{\cal X}({\bold y})}{\text{{argmin}}}\left\Vert \mathbf{x}-
\mathbf{u}\right\Vert .\label{eq:Eucl_projection}
\end{equation}

\begin{lemma}\label{th:Lemma_descent} The best-response function
$\mathcal{X}\ni\mathbf{y}\mapsto\hat{\bold x}(\mathbf{y})$ satisfies
the following: 

\noindent(i) For every $\mathbf{y}\in\mathcal{X}$, $\hat{\bold x}(\mathbf{y})-\mathbf{y}$
is a descent direction for $U$ at $\mathbf{y}$ such that
\begin{equation}
\nabla U(\mathbf{y})\trt(\hat{\bold x}(\mathbf{y})-\mathbf{y})\le-c_{\tilde{U}}\|\hat{\bold x}(\mathbf{y})-\mathbf{y}\|^{2},\label{eq:suffdescent-1}
\end{equation}
where $c_{\tilde{U}}>0$ is the constant of uniform strong convexity
of $\tilde{{U}}$ (cf. B1);

\noindent(ii) For every $\mathbf{y}\in\mathcal{X}$, it holds that
\begin{equation}
\hat{\bold x}({\bold y})=P_{{\cal X}({\bold y})}\left(\mathbf{w}_{\rho}\left(\hat{\bold x}({\bold y}),{\bold y}\right)\right),\label{eq:optproj}
\end{equation}
for every fixed  $\rho>0$.

\noindent(iii) Suppose that also B5 holds true. Then, $\hat{\bold x}(\bullet)$
is continuous at every $\bar{{\mathbf{x}}}\in\mathcal{X}$ such that
$\hat{\bold x}(\bar{{\mathbf{x}}})\in\mathcal{X}(\bar{{\mathbf{x}}})$
is regular.\end{lemma}

\begin{proof} (i) By Assumption 2, for any given $\mathbf{y}\in{\cal {\cal X}}$,
$\hat{\bold x}(\mathbf{y})$ is the solution of the strongly convex
optimization problem \ref{eq:k2}; therefore,
\begin{equation}
({\bold z}-\hat{\bold x}(\mathbf{y}))\trt\nabla_{\bold x}\tilde{U}\left(\hat{\bold x}(\mathbf{y});\mathbf{y}\right)\ge0,\quad\forall{\bold z}\in{\cal X}(\mathbf{y}).\label{eq:minprin}
\end{equation}
By choosing ${\bold z}=\mathbf{y}$ {[}recall by Lemma \ref{th:uspropfs-1}(i) that $\mathbf{y}\in{\cal {\cal X}}(\mathbf{y})${]},
we get
\[
\begin{array}{ll}
\left(\mathbf{y}-\hat{\bold x}(\mathbf{y})\right)\trt\left(\nabla_{\bold x}\tilde{U}(\hat{\bold x}(\mathbf{y});\mathbf{y})\right.\hspace{-8pt} & -\nabla_{\bold x}\tilde{U}(\mathbf{y};\mathbf{y})\\[5pt]
 & \left.+\nabla_{\bold x}\tilde{U}(\mathbf{y};\mathbf{y})\right)\ge0,
\end{array}
\]
which, using B1 and B2, leads to
\[
({\bold y}-\hat{\bold x}(\mathbf{y}))\trt\nabla_{\bold x}U(\mathbf{y})\ge c_{\tilde{U}}\|\hat{\bold x}(\mathbf{y})-\mathbf{y}\|^{2}.
\]

\noindent(ii) It follows readily from the fixed-point characterization
of the solution $\hat{\bold x}({\bold y})$ of the strongly convex
subproblem $\mathcal P_\by$: see, e.g., \cite[Prop. 1.5.8]{FacchPangBk}.

\noindent(iii) We first observe that, under the assumed regularity of all the points in $\mathcal{X}\!\left(\mathbf{\bar{{\mathbf{x}}}}\right)$,
$\mathcal{X}(\mathbf{\bullet})$ is continuous at $\bar{{\mathbf{x}}}$
\cite[Example 5.10]{RockWets98}. It follows from \cite[Proposition 4.9]{RockWets98}
(see also \cite[Example 5.57]{RockWets98}) that, for every fixed $\mathbf{u}\in\mathbb R^n$,
the map $\mathbf{x}\mapsto P_{{\cal X}({\bold x})}({\bold u})$ is
continuous at $\mathbf{x}=\bar{\mathbf{x}}$. This, together with B1,
B3 and B5 is sufficient for $\hat{\bold x}(\bullet)$ to be continuous
at $\bar{{\mathbf{x}}}$ \cite[Theorem 2.1]{Daf88}.\end{proof}\medskip

\noindent \noindent \textbf{Lemma \ref{th:uspropofupgpre}}. Under the extra
conditions B4-B5, with the following lemma, which is reminiscent of similar results in \cite{Daf88} and \cite{yen95}, we can establish a suitable sensitivity property
of the best-response function $\hat{\bold x}(\bullet)$;
Lemma \ref{th:uspropofupgpre} will play a key role in the proof of
statement (c) of Theorem \ref{th:conver}.

\begin{lemma}\label{th:uspropofupgpre} Suppose 
that B4-B5 hold and there exist $\bar{\rho}>0$ and $\beta>0$
such that
\begin{equation}
\|P_{{\cal X}({\bold y})}(\mathbf{w}_{\rho}(\hat{\bold x}({\bold z}),{\bold z}))-P_{{\cal X}({\bold z})}(\mathbf{w}_{\rho}(\hat{\bold x}({\bold z}),{\bold z}))\|\le\beta\|{\bold y}-{\bold z}\|^{\frac{1}{2}},\label{eq:projpro}
\end{equation}
for all $\rho\in(0,\bar{\rho}]$ and ${\bold y},\,{\bold z}\,\in{\cal X}$.
Then there exists $\tilde{\rho}\in(0,\bar{\rho}]$ such that
\begin{equation}
\|\hat{\bold x}({\bold y})-\hat{\bold x}({\bold z})\|\le\eta_{\rho}\|{\bold y}-{\bold z}\|+\theta_{\rho}\|{\bold y}-{\bold z}\|^{\frac{1}{2}},\label{eq:hol}
\end{equation}
for all ${\bold y},{\bold z}\in{\cal X}$ and $\rho\in(0,\tilde{\rho}]$,
with
\begin{equation}
\begin{array}{lll}
\eta_{\rho} & \triangleq & \dfrac{\rho\,\tilde{L}_{\nabla,2}}{1-\sqrt{1+\rho^{2}\tilde{L}_{\nabla,1}^{2}-2\rho c_{\tilde{U}}}}\\
\theta_{\rho} & \triangleq & \dfrac{{\beta}}{1-\sqrt{1+\rho^{2}\tilde{L}_{\nabla,1}^{2}-2\rho c_{\tilde{U}}}},
\end{array}\label{eq:def_eta_theta}
\end{equation}
 where $\tilde{L}_{\nabla,1}$ and $\tilde{L}_{\nabla,2}$ are the
Lipschitz constants of $\nabla_{\bold x}\tilde{U}(\bullet;\mathbf{y})$
and $\nabla_{\bold x}\tilde{U}(\mathbf{x};\bullet)$, respectively
(cf. B4 and B5); $\tilde{L}_{\nabla,1}$ is assumed to be such that $\tilde{L}_{\nabla,1}\geq c_{\tilde{U}}$
without loss of generality.\end{lemma}
\begin{IEEEproof}
Using \eqref{eq:optproj} we have, for every $\rho>0$,
\begin{equation}
\begin{array}{l}
\|\hat{\bold x}({\bold y})-\hat{\bold x}({\bold z})\|\smallskip\\
\quad=\|P_{{\cal X}({\bold y})}(\mathbf{w}_{\rho}(\hat{\bold x}({\bold y}),{\bold y}))-P_{{\cal X}({\bold z})}(\mathbf{w}_{\rho}(\hat{\bold x}({\bold z}),{\bold z}))\|\smallskip\\
\quad\leq\|P_{{\cal X}({\bold y})}(\mathbf{w}_{\rho}(\hat{\bold x}({\bold y}),{\bold y}))-P_{{\cal X}({\bold y})}(\mathbf{w}_{\rho}(\hat{\bold x}({\bold z}),{\bold y}))\|\smallskip\\
\quad\quad+\|P_{{\cal X}({\bold y})}(\mathbf{w}_{\rho}(\hat{\bold x}({\bold z}),{\bold y}))-P_{{\cal X}({\bold z})}(\mathbf{w}_{\rho}(\hat{\bold x}({\bold z}),{\bold z})\|.
\end{array}\label{eq:ineq}
\end{equation}
We bound next the two terms on the RHS of (\ref{eq:ineq}).

For every $\rho>0$, it holds
\begin{equation}
\!\!\!\!\begin{array}{ll}
 & \!\!\!\!\!\!\|P_{{\cal X}({\bold y})}(\mathbf{w}_{\rho}(\hat{\bold x}({\bold y}),{\bold y}))-P_{{\cal X}({\bold y})}(\mathbf{w}_{\rho}(\hat{\bold x}({\bold z}),{\bold y}))\|^{2}\\[5pt]
 & \quad\overset{(a)}{\leq}\|\mathbf{w}_{\rho}(\hat{\bold x}({\bold y}),{\bold y})-\mathbf{w}_{\rho}(\hat{\bold x}({\bold z}),{\bold y})\|^{2}\\[5pt]
 & \quad=\,\|\hat{\bold x}({\bold y})-\hat{\bold x}({\bold z})\|^{2}\medskip\\
 & \quad\quad\,+\rho^{2}\|\nabla_{\bold x}\tilde{U}(\hat{\bold x}({\bold z});{\bold y})-\nabla_{\bold x}\tilde{U}(\hat{\bold x}({\bold y});{\bold y})\|^{2}\medskip\\
 & \quad\hspace{5pt}\hspace{5pt}\,-2\rho\,\left(\hat{\bold x}({\bold y})-\hat{\bold x}({\bold z})\right)\trt\left(\nabla_{\bold x}\tilde{U}(\hat{\bold x}({\bold y});{\bold y})-\nabla_{\bold x}\tilde{U}(\hat{\bold x}({\bold z});{\bold y})\right)\\[5pt]
 & \quad\overset{(b)}{\leq}\|\hat{\bold x}({\bold y})-\hat{\bold x}({\bold z})\|^{2}+\rho^{2}\tilde{L}_{\nabla,1}^{2}\|\hat{\bold x}({\bold y})-\hat{\bold x}({\bold z})\|^{2}\\[5pt]
 & \quad\hspace{5pt}\hspace{5pt}\,-2\,\rho\, c_{\tilde{U}}\,\|\hat{\bold x}({\bold y})-\hat{\bold x}({\bold z})\|^{2}\\[5pt]
 & \quad=\,(1+\rho^{2}\tilde{L}_{\nabla,1}^{2}-2\,\rho c_{\tilde{U}})\,\|\hat{\bold x}({\bold y})-\hat{\bold x}({\bold z})\|^{2},\\[5pt]
\end{array}\label{eq:preresstron}
\end{equation}
where (a) is due to the non-expansive property of the projection
operator $P_{{\cal X}({\bold y})}(\bullet)$ and (b) follows from
B1 and B5. Note that  $1+\rho^{2}\tilde{L}_{\nabla,1}^{2}-2\rho c_{\tilde{U}}>0$ since we assumed $\tilde{L}_{\nabla,1}\geq c_{\tilde{U}}$.

Let us bound now the second term on the RHS of (\ref{eq:ineq}). For
every $\rho\in(0,\bar{\rho}]$, we have \vspace{-0.2cm}

\begin{equation}
\begin{array}{l}
\!\!\!\!\!\!\!\|P_{{\cal X}({\bold y})}(\mathbf{w}_{\rho}(\hat{\bold x}({\bold z}),{\bold y}))-P_{{\cal X}({\bold z})}(\mathbf{w}_{\rho}(\hat{\bold x}({\bold z}),{\bold z}))\|\medskip\\[5pt]
\quad\quad\le\|P_{{\cal X}({\bold y})}(\mathbf{w}_{\rho}(\hat{\bold x}({\bold z}),{\bold y}))-P_{{\cal X}({\bold y})}(\mathbf{w}_{\rho}(\hat{\bold x}({\bold z}),{\bold z}))\|\\[5pt]
\quad\quad\hspace{5pt}\hspace{5pt}+\|P_{{\cal X}({\bold y})}(\mathbf{w}_{\rho}(\hat{\bold x}({\bold z}),{\bold z}))-P_{{\cal X}({\bold z})}(\mathbf{w}_{\rho}(\hat{\bold x}({\bold z}),{\bold z}))\|\\[5pt]
\quad\quad\overset{(a)}{\leq}\|\mathbf{w}_{\rho}(\hat{\bold x}({\bold z}),{\bold y})-\mathbf{w}_{\rho}(\hat{\bold x}({\bold z}),{\bold z})\|+\beta\|{\bold y}-{\bold z}\|^{\frac{1}{2}}\\[5pt]
\quad\quad\overset{(b)}{\leq}\rho\,\tilde{L}_{\nabla,2}\,\|{\bold y}-{\bold z}\|+\beta\|{\bold y}-{\bold z}\|^{\frac{1}{2}},
\end{array}\label{eq:intpass2}
\end{equation}
where (a) is due to the non-expansive property of the projection
$P_{{\cal X}({\bold y})}(\bullet)$ and \eqref{eq:projpro}, and (b)
follows from B4.

Combining \eqref{eq:ineq}, \eqref{eq:preresstron} and \eqref{eq:intpass2}
we obtain the desired result \eqref{eq:hol} with $\tilde{\rho}=\min\{2c_{\tilde{U}}/\tilde{L}_{\nabla,1}^{2},\bar{\rho}\}$
(so that $0< 1+\rho^{2}\tilde{L}_{\nabla,1}^{2}-2\rho c_{\tilde{U}}<1$
for every $\rho\in(0,\tilde{\rho}]$).\medskip
\end{IEEEproof}
\noindent \noindent \textbf{Lemmas \ref{th:preyen} and \ref{th:basicarg}}.
While Assumptions 1-3 and B4-B5 in Lemma \ref{th:uspropofupgpre}
are quite standard, condition \eqref{eq:projpro} is less trivial
and not easy to be checked. The following Lemma \ref{th:basicarg} provides
some easier to be checked sufficient conditions that imply \eqref{eq:projpro}.
To prove Lemma \ref{th:basicarg} we need first Lemma \ref{th:preyen},
as stated next.

\begin{lemma}\label{th:preyen} Consider $\bar \bx \in {\cal X}$. By assuming C7, the following hold:

\noindent(i) If $\tilde{{\mathbf{x}}}\in{\cal X}(\bar{\bold x})$
is regular, then ${\cal X}(\bullet)$
enjoys the Aubin property  at $(\bar{\bold x},\tilde{{\mathbf{x}}})$;%
\footnote{See \cite[Def. 9.36]{RockWets98} for the definition of the Aubin property.
Note also that we use some results from  \cite{yen95}  where a point-to-set map that has the Aubin property
is called pseudo-Lipschitz
\cite[Def. 1.1]{yen95}.
}

\noindent(ii) If in addition ${\cal X}$
is compact, then a neighborhood $\mathcal{V}_{\bar{\bold x}}$ of
$\bar{\bold x}$, a neighborhood $\mathcal{W}_{\tilde{{\mathbf{x}}}}$
of $\tilde{{\mathbf{x}}}$, and a constant $\hat{\beta}>0$ exist
such that
\begin{equation}
\|P_{{\cal X}({\bold y})}({\bold u})-P_{{\cal X}({\bold z})}({\bold u})\|\le\hat{\beta}\|{\bold y}-{\bold z}\|^{\frac{1}{2}}\label{eq:pseudlip}
\end{equation}
for all ${\bold y},{\bold z}\in{\cal X}\cap\mathcal{V}_{\bar{\bold x}}$,
and ${\bold u}\in\mathcal{W}_{\tilde{{\mathbf{x}}}}$. \end{lemma}
\begin{IEEEproof}
(i) Under Assumptions 1-3 and C7, the statement follows readily from
 \cite[Theorem 3.2]{Rock85} in
view of the regularity of $\tilde{{\mathbf{x}}}$.

\noindent(ii) Since ${\cal X}(\bullet)$ has the Aubin property  at $(\bar{\bold x},\tilde{{\mathbf{x}}})$,
there exist a neighborhood $\mathcal{V}_{\bar{\bold x}}$ of $\bar{\bold x}$,
a neighborhood $\mathcal{W}_{\tilde{{\mathbf{x}}}}$ of $\tilde{{\mathbf{x}}}$,
and a constant $\hat{\beta}>0$ such that \cite[ Lemma 1.1]{yen95}:
\begin{equation}
\begin{array}{l}
\|P_{{\cal X}({\bold y})\cap\mathcal{B}_{\tilde{{\mathbf{x}}}}}({\bold u})-P_{{\cal X}({\bold z})\cap\mathcal{B}_{\tilde{{\mathbf{x}}}}}({\bold u})\|\le\hat{\beta}\|{\bold y}-{\bold z}\|^{\frac{1}{2}}\\[5pt]\end{array},\label{eq:pseudlippr}
\end{equation}
for all ${\bold y},{\bold z}\in{\cal X}\cap\mathcal{V}_{\bar{\bold x}}$,
and ${\bold u}\in\mathcal{W}_{\tilde{{\mathbf{x}}}}$, where $\mathcal{B}_{\tilde{{\mathbf{x}}}}$
denotes a closed convex neighborhood of $\tilde{{\mathbf{x}}}$. Since
${\cal X}$ is compact, one can always choose $\mathcal{B}_{\tilde{{\mathbf{x}}}}$
such that ${\cal X}(\bar{\bold x})\subset\mathcal{B}_{\tilde{{\mathbf{x}}}}$ for every $\bar \bx \in \mathcal{X}$ and, thus,
\[
\|P_{{\cal X}({\bold y})\cap\mathcal{B}_{\tilde{{\mathbf{x}}}}}({\bold u})-P_{{\cal X}({\bold z})\cap\mathcal{B}_{\tilde{{\mathbf{x}}}}}({\bold u})\|=\|P_{{\cal X}({\bold y})}({\bold u})-P_{{\cal X}({\bold z})}({\bold u})\|,
\]
which proves the desired result.\medskip
\end{IEEEproof}
\noindent We can now derive sufficient conditions for \eqref{eq:projpro}
to hold. \vspace{-0.2cm}

\noindent \begin{lemma}\label{th:basicarg} Suppose that C7 holds
true, $\mathcal{X}$ is compact and $\hat{\bold x}(\bar{\bold x})\in{\cal X}(\bar{\bold x})$
is regular for every $\bar{\bold x}\in{\cal X}$. Then, property
\eqref{eq:projpro} holds. \end{lemma}
\begin{IEEEproof}
It follows from Lemma \ref{th:preyen}(ii) that, for every $\bar{\bold x}\in{\cal X}$, there exist 
a neighborhood $\mathcal{V}_{\bar{\bold x}}$ of $\bar{\bold x}$,
a neighborhood $\mathcal{W}_{\hat{\bold x}(\bar{\bold x})}$ of $\hat{\bold x}(\bar{\bold x})$,
and a constant $\hat{\beta}>0$  such that:
\begin{equation}
\|P_{{\cal X}({\bold y})}({\bold u})-P_{{\cal X}({\bold z})}({\bold u})\|\le\hat{\beta}\|{\bold y}-{\bold z}\|^{\frac{1}{2}}\label{eq:pseudlipdim}
\end{equation}
for every ${\bold y},{\bold z}\in{\cal X}\cap\mathcal{V}_{\bar{\bold x}},\,{\bold u}\in\mathcal{W}_{\hat{\bold x}(\bar{\bold x})}$.

Suppose now by contradiction that \eqref{eq:projpro} does not hold.
Then, for all $\bar{\rho}^{\nu}>0$ and $\beta^{\nu}>0$ there exist
$\rho^{\nu}\in(0,\bar{\rho}^{\nu}]$, $\bar{\bold x}^{\nu}$, and
${\bar{{\bold y}}}^{\nu}\,\in{\cal X}$ such that:
\begin{equation}
\!\!\!\!\begin{array}{l}
\|P_{{\cal X}({\bar{{\bold y}}}^{\nu})}(\mathbf{w}_{\rho^{\nu}}(\hat{\bold x}\left(\bar{\bold x}^{\nu}\right),\bar{\bold x}^{\nu}))-P_{{\cal X}(\bar{\bold x}^{\nu})}(\mathbf{w}_{\rho^{\nu}}(\hat{\bold x}\left(\bar{\bold x}^{\nu}\right),\bar{\bold x}^{\nu}))\|\\
\hfill>\beta^{\nu}\|{\bar{{\bold y}}}^{\nu}-\bar{\bold x}^{\nu}\|^{\frac{1}{2}}.
\end{array}\label{eq:contr}
\end{equation}
Furthermore, in view of the compactness of ${\cal X}$, denoting by
$D_{{\cal X}}$ the (finite) diameter of ${\cal X}$, the LHS of (\ref{eq:contr})
can be bounded by
\begin{equation}
D_{{\cal X}}\ge\|P_{{\cal X}(\bar{\bold y}^{\nu})}(\mathbf{w}_{\rho^{\nu}}(\hat{\bold x}\left(\bar{\bold x}^{\nu}\right),\bar{\bold x}^{\nu}))-P_{{\cal X}(\bar{\bold x}^{\nu})}(\mathbf{w}_{\mathbf{\rho^{\nu}}}(\hat{\bold x}\left(\bar{\bold x}^{\nu}\right),\bar{\bold x}^{\nu}))\|.\label{eq:contrpass}
\end{equation}

Suppose without loss of generality that $\beta^{\nu}\to+\infty$,
$\bar{\rho}^{\nu}\downarrow0$, and $\bar{\bold x}^{\nu}\underset{{\cal N}}{\to}\bar{\bold x}\in{\cal X}(\bar{\bold x})\subseteq{\cal X}$
and ${\bar{{\bold y}}}^{\nu}\underset{{\cal N}}{\to}\bar{\bold y}\in{\cal X}(\bar{\bold y})\subseteq{\cal X}$,
possibly on a suitable subsequence ${\cal N}$ {[}recall that $\bar{\bold x}^{\nu}\in{\cal X}(\bar{\bold x}^{\nu})$
and ${\bar{{\bold y}}}^{\nu}\in{\cal X}({\bar{{\bold y}}}^{\nu})${]}.
From \eqref{eq:contr} and \eqref{eq:contrpass}, we obtain
\[
D_{{\cal X}}\ge\limsup_{\nu\to+\infty}\beta^{\nu}\|\bar{\bold y}^{\nu}-\bar{\bold x}^{\nu}\|^{\frac{1}{2}},
\]
which, in turn, considering that $\beta^{\nu}\to\infty$ and $\|\bar{\bold y}^{\nu}-\bar{\bold x}^{\nu}\|^{\frac{1}{2}}\ge0$,
implies
\begin{equation}
\lim_{\nu\to+\infty}\|\bar{\bold y}^{\nu}-\bar{\bold x}^{\nu}\|^{\frac{1}{2}}=0.\label{eq:equal_limit_point}
\end{equation}
Then, it must be $\bar{\bold x}=\bar{\bold y}$.

Invoking now the continuity of $\hat{\bold x}(\bullet)$ at $\bar{{\mathbf{x}}}$
{[}cf. Lemma \ref{th:Lemma_descent}(iii){]} and $\nabla_{\bold x}\tilde{U}(\bullet;\bullet)$
on $\mathcal{K}\times\mathcal{X}$ {[}cf. B3{]}, we have
\begin{equation}
\mathbf{w}_{\rho^{\nu}}\!\left(\hat{\bold x}\left(\bar{\bold x}^{\nu}\right),\bar{\bold x}^{\nu}\right)=\hat{\bold x}\left(\bar{\bold x}^{\nu}\right)-\rho^{\nu}\nabla_{\bold x}\tilde{U}(\hat{\bold x}\left(\bar{\bold x}^{\nu}\right),\bar{\bold x}^{\nu})\underset{{\cal N}}{\to}\hat{\bold x}(\bar{\bold x}).\label{eq:w_rho_convergence}
\end{equation}

Therefore, for every $\hat{\beta}>0$ and neighborhoods $\mathcal{V}_{\bar{{\mathbf{x}}}}$
and $\mathcal{W}_{\hat{\bold x}(\bar{\bold z})}$, there exists a
sufficiently large $\nu$ such that \eqref{eq:contr} holds with $\beta^{\nu}>\hat{\beta}$
(recall that $\beta^{\nu}\rightarrow+\infty$), $\bar{{\mathbf{x}}}^{\nu},\,\bar{{\mathbf{y}}}^{\nu}\in\mathcal{V}_{\bar{{\mathbf{x}}}} \cap \mathcal{X}$
{[}due to (\ref{eq:equal_limit_point}){]}, and $\mathbf{w}_{\rho^{\nu}}\!\left(\hat{\bold x}\left(\bar{\bold x}^{\nu}\right),\bar{\bold x}^{\nu}\right)\in\mathcal{W}_{\hat{\bold x}(\bar{\bold z})}$
{[}due to (\ref{eq:w_rho_convergence}){]}; this is in contradiction
with \eqref{eq:pseudlipdim}.
\end{IEEEproof}
\smallskip{}
We recall that the assumption on the regularity of $\hat \bx(\bar \bx) \in {\mathcal{X}(\bar \bx)}$ for every $\bar \bx \in \mathcal{X}$, as required in Lemma \ref{th:basicarg}, is  implied by Assumption 5.
\medskip{}

\noindent \textbf{Theorem \ref{th:conver-1}. }The last theorem of
this section provides technical conditions under which a(t least one)
regular limit point of the sequence generated by Algorithm 1 is a
stationary solution of the original nonconvex problem \ref{eq:nonnonU}.

\begin{theorem}\label{th:conver-1} Let $\{{\bold x}^{\nu}\}$ be
the sequence generated by Algorithm 1 under Assumptions 1-3 and 5. The
the following hold. 

\noindent(a) Suppose
\begin{equation}
\underset{_{\nu\rightarrow\infty}}{\text{{liminf}}}\,\|\hat{\bold x}({\bold x}^{\nu})-{\bold x}^{\nu}\|=0.\label{eq:Th_preliminary_subsequence_convergence}
\end{equation}
 Then, at least one regular limit point of $\{{\bold x}^{\nu}\}$
is a stationary solution of \ref{eq:nonnonU}.

\noindent(b) Suppose
\begin{equation}
\lim_{\nu\rightarrow\infty}\|\hat{\bold x}({\bold x}^{\nu})-{\bold x}^{\nu}\|=0.\label{eq:Th_preliminary_sequence_convergence}
\end{equation}
 Then, every regular limit point of $\{{\bold x}^{\nu}\}$ is a stationary
solution of \ref{eq:nonnonU}.\end{theorem}\begin{proof} We prove
only (a); (b) follows applying the result in (a) to every convergent
subsequence of $\{{\bold x}^{\nu}\}$.

Let $\bar{\bold x}$ be a regular accumulation point of the subsequence
$\{{\bold x}^{\nu}\}_{{\cal N}}$ of $\{{\bold x}^{\nu}\}$ satisfying
\eqref{eq:Th_preliminary_subsequence_convergence}; thus, there exists
${\cal N'}\subseteq{\cal N}$ such that $\lim_{{\cal N}^{'}\ni\nu\rightarrow\infty}{\bold x}^{\nu}=\bar{\bold x}$.
We show next that $\bar{\bold x}$ is a KKT point of the original
problem. Let $\bar{J}$ and $J^{\nu}$ be the following sets:
\[
\bar{J}\triangleq\{j\in[1,\ldots,m]:\, g_{j}(\bar{\bold x})=0\},
\]
\[
J^{\nu}\triangleq\{j\in[1,\ldots,m]:\,\tilde{g}_{j}(\hat{\bold x}({\bold x}^{\nu});{\bold x}^{\nu})=0\}
\]
with $\nu\in{\cal N'}$. Using $\lim_{{\cal N}^{'}\ni\nu\rightarrow\infty}\|\hat{\bold x}({\bold x}^{\nu})-{\bold x}^{\nu}\|=0$
{[}cf. \eqref{eq:Th_preliminary_subsequence_convergence}{]} along
with the continuity of $\tilde{g}_{j}$, by C2, we have
\begin{equation}
\lim_{{\cal N}^{'}\ni\nu\rightarrow\infty}\tilde{g}_{j}(\hat{\bold x}({\bold x}^{\nu});{\bold x}^{\nu})=\tilde{g}_{j}(\bar{\bold x};\bar{\bold x})=g_{j}(\bar{\bold x}),\; j=1,\ldots,m.\label{eq:convlim}
\end{equation}

The limit above implies that there exists a positive integer $\tilde{\nu}\in{\cal N'}$
such that
\begin{equation}
J^{\nu}\subseteq\bar{J},\;\forall\nu\ge\tilde{\nu}\,\,\mbox{and}\,\,\nu\in{\cal N'}.\label{eq:actsets}
\end{equation}
Since the functions $\nabla_{\bold x}\tilde U$ and $\nabla_{\bold x}\tilde{g}_{j}$
are continuous, we get, by B2,
\begin{equation}
\lim_{{\cal N}^{'}\ni\nu\rightarrow\infty}\nabla_{\bold x}\tilde{U}(\hat{\bold x}({\bold x}^{\nu});{\bold x}^{\nu})=\nabla_{\bold x}\tilde{U}(\bar{\bold x};\bar{\bold x})=\nabla U(\bar{\bold x}),\label{eq:gradoflim}
\end{equation}
and, for $j=1,\ldots,m$, by C5,
\begin{equation}
\lim_{{\cal N}^{'}\ni\nu\rightarrow\infty}\nabla_{\bold x}\tilde{g}_{j}(\hat{\bold x}({\bold x}^{\nu});{\bold x}^{\nu})=\nabla_{\bold x}\tilde{g}_{j}(\bar{\bold x};\bar{\bold x})=\nabla g_{j}(\bar{\bold x}).\label{eq:gradconlim}
\end{equation}
We claim now that for sufficiently large $\nu\in{\cal N'}$, the MFCQ
holds at $\hat{\bold x}({\bold x}^{\nu})\in{\cal X}({\bold x}^{\nu})$.
Assume by contradiction that the following implication does not hold
for infinitely many $\nu\in{\cal N'}$:
\begin{equation}
\begin{array}{c}
-\sum_{j\in J^{\nu}}\mu_{j}^{\nu}\nabla_{\bold x}\tilde{g}_{j}(\hat{\bold x}({\bold x}^{\nu});{\bold x}^{\nu})\in N_{{\cal K}}(\hat{\bold x}({\bold x}^{\nu}))\\[5pt]
\underbrace{\mbox{}\hspace{70pt}\mu_{j}^{\nu}\ge0,\;\forall j\,\in\, J^{\nu},\hspace{70pt}}\\[-2pt]
\Downarrow\\
\mu_{j}^{\nu}=0,\;\forall j\in\, J^{\nu}.
\end{array}\label{eq:MFCQk}
\end{equation}
It follows that a nonempty index set $\bar{\bar{J}}\subseteq\bar{J}$
exists such that, after a suitable renumeration, for every $\nu\in{\cal N'}$,
we must have
\begin{equation}
\begin{array}{c}
-\sum_{j\in\bar{\bar{J}}}\mu_{j}^{\nu}\nabla_{\bold x}\tilde{g}_{j}(\hat{\bold x}({\bold x}^{\nu});{\bold x}^{\nu})\in N_{{\cal K}}\left(\hat{\bold x}({\bold x}^{\nu})\right)\\[5pt]
\mu_{j}^{\nu}\ge0,\,\forall j\in\bar{\bar{J}}\\[5pt]
\sum_{j\in\bar{\bar{J}}}\mu_{j}^{\nu}=1.
\end{array}\label{eq:limmfcq}
\end{equation}
\textcolor{black}{We may assume without loss of generality that, for each $j\in\bar{\bar{J}}$,
the sequence $\{\mu_{j}^{\nu}\}$ converges to a limit $\bar{\mu}_{j}$
such that $\sum_{j\in\bar{\bar{J}}}\bar{\mu}_{j}=1.$}\textcolor{red}{{}
}In view of the inclusion $\bar{\bar{J}}\subseteq\bar{J}$, by taking
the limit ${\cal N}^{'}\ni\nu\rightarrow\infty$ in \eqref{eq:limmfcq},
and invoking  \eqref{eq:gradconlim} along with the outer semicontinuity
of the mapping $N_{{\cal K}}(\bullet)$ \cite[Prop. 6.6]{RockWets98},
we get
\begin{equation}
\begin{array}{c}
-\sum_{j\in\bar{\bar{J}}}\bar{\mu}_{j}\nabla_{\bold x}g_{j}(\bar{\bold x})\in N_{{\cal K}}\left(\bar{\bold x}\right)\\[5pt]
\bar{\mu}_{j}\ge0,\,\forall j\in\bar{\bar{J}}\\[5pt]
\sum_{j\in\bar{\bar{J}}}\bar{\mu}_{j}=1,
\end{array}\label{eq:limmfcq-limit}
\end{equation}
in contradiction with the regularity of $\bar{\bold x}$ {[}the
MFCQ holds at $\bar{\bold x}$, see (\ref{eq:MFCQ-2}){]}. Therefore,
(\ref{eq:MFCQk}) must hold for sufficiently large $\nu\in{\cal N'}$,
implying that the KKT system of problem  \ref{eq:k2} has a solution
for every sufficiently large $\nu\in{\cal N'}$: thus, there exist
$(\mu_{j}^{\nu})_{j=1}^{m}$ such that \vspace{-0.2cm}

\begin{equation}
\hspace{-0.4cm}\begin{array}{c}
-\left[\nabla_{\bold x}\tilde{U}(\hat{\bold x}({\bold x}^{\nu});{\bold x}^{\nu})+\sum_{j=1}^{m}\mu_{j}^{\nu}\nabla_{\bold x}\tilde{g}_{j}(\hat{\bold x}({\bold x}^{\nu});{\bold x}^{\nu})\right]\!\in\! N_{{\cal K}}(\hat{\bold x}({\bold x}^{\nu}))\\[5pt]
0\le\mu_{j}^{\nu}\perp\tilde{g}_{j}(\hat{\bold x}({\bold x}^{\nu});{\bold x}^{\nu})\le0,\enspace j=1,\ldots,m.
\end{array}\label{eq:KKTk}
\end{equation}
Note that by (\ref{eq:actsets}) and the complementarity slackness
in (\ref{eq:KKTk}), $\mu_{j}^{\nu}=0$ for all $j\notin\bar{J}$
and large $\nu\in{\cal N'}$. \textcolor{black}{Moreover, the sequence
of nonnegative multipliers $\{\boldsymbol{{\mu}}^{\nu}\triangleq(\mu_{j}^{\nu})_{j\in\bar{J}}\}_{\nu\in{\cal N'}}$
must} be bounded, as shown next. Suppose by contradiction that $\lim_{{\cal N'}\ni\nu\to\infty}\|\boldsymbol{{\mu}}^{\nu}\|=+\infty$
for some $\{\hat{\bold x}({\bold x}^{\nu})\}_{{\cal N'}}$ (possibly
over a subsequence). Dividing both sides of \eqref{eq:KKTk} by $\|\boldsymbol{{\mu}}^{\nu}\|$
and taking the limit ${\cal N}^{'}\ni\nu\rightarrow\infty$, one would
get
\begin{equation}
\begin{array}{c}
-\sum_{j\in\bar{J}}\bar{\bar{\mu}}_{j}\nabla g_{j}(\bar{\bold x})\in N_{{\cal K}}(\bar{\bold x})\\[5pt]
0\le\bar{\bar{\mu}}_{j}\perp g_{j}(\bar{\bold x})\le0,\enspace j\in\bar{J},
\end{array}\label{eq:contradiction_bounded_multipliers}
\end{equation}
for some $\bar{\bar{\bm{\mu}}}\triangleq(\bar{\bar{\mu}}_{j})_{j\in\bar{J}}\neq\mathbf{0}$,
in contradiction with \textcolor{black}{{} (\ref{eq:MFCQ-2}).}

Therefore, \textcolor{black}{$\{\boldsymbol{{\mu}}^{\nu}\triangleq(\mu_{j}^{\nu})_{j\in\bar{J}}\}_{\nu\in{\cal N'}}$}
must have a limit; let us denote by $(\bar{\mu}_{j})_{j\in\bar{J}}$
such a limit (after a suitable renumeration). Taking the limit ${\cal N'}\ni\nu\to\infty$
in \eqref{eq:KKTk}, and using  \eqref{eq:gradoflim} and \eqref{eq:gradconlim} along with the
outer semicontinuity of the mapping $N_{{\cal K}}(\bullet)$, we get
\begin{equation}
\begin{array}{c}
-\left[\nabla U(\bar{\bold x})+\sum_{j\in\bar{J}}\bar{\mu}_{j}\nabla g_{j}(\bar{\bold x})\right]\in N_{{\cal K}}(\bar{\bold x})\\[5pt]
0\le\bar{\mu}_{j}\perp g_{j}(\bar{\bold x})\le0,\quad j\in\bar{J}.
\end{array}\label{eq:KKT_original_problem}
\end{equation}
It follows from \eqref{eq:KKT_original_problem} that $\bar{\bold x}$
is a stationary solution of the original problem \ref{eq:nonnonU}.
\end{proof}

\subsection{Proof of Theorem \ref{th:conver} \label{sub:Proof_main_theo}}

\noindent\textbf{Proof of statement (a).} It follows from Lemma \ref{th:uspropfs-1}.\smallskip{}

\noindent\textbf{Proof of statement (b).} Invoking Theorem \ref{th:conver-1}(b),
it is sufficient to show that
(\ref{eq:Th_preliminary_sequence_convergence}) in Theorem \ref{th:conver-1}
is satisfied.

By the descent lemma \cite[Propo. A.24]{Bertsekas_Book-Parallel-Comp} and Step
3 of Algorithm 1, we get:
\[
\begin{array}{rl}
U({\bold x}^{\nu+1})\le & U({\bold x}^{\nu})+\gamma^{\nu}\nabla U({\bold x}^{\nu})\trt(\hat{\bold x}({\bold x}^{\nu})-{\bold x}^{\nu})\\[5pt]
 & +\frac{(\gamma^{\nu})^{2}L_{\nabla U}}{2}\|\hat{\bold x}({\bold x}^{\nu})-{\bold x}^{\nu}\|^{2}.
\end{array}
\]
Invoking (\ref{eq:suffdescent-1}) in Lemma \ref{th:Lemma_descent}, we obtain
\begin{equation}
U({\bold x}^{\nu+1})\le U({\bold x}^{\nu})-\gamma^{\nu}\left(c_{\tilde{U}}-\frac{\gamma^{\nu}L_{\nabla U}}{2}\right)\|\hat{\bold x}({\bold x}^{\nu})-{\bold x}^{\nu}\|^{2}.\label{eq:desclem}
\end{equation}
Since $0<\inf_{\nu}\gamma^{\nu}\le\sup_{\nu}\gamma^{\nu}\le\gamma^{\max}\le1$
and $2c_{\tilde{U}}>\gamma^{\max}L_{\nabla U}$, we deduce from \eqref{eq:desclem}
that either ${U({\bold x}^{\nu})}\to-\infty$ or  $\{U({\bold x}^{\nu})\}$
converges to a finite value and
\begin{equation}
\lim_{\nu \to \infty}\|\hat{\bold x}({\bold x}^{\nu})-{\bold x}^{\nu}\|=0.\label{eq:distconv}
\end{equation}
By assumption A4, $\{U({\bold x}^{\nu})\}$
is convergent and the
sequence $\{\mathbf{x}^{\nu}\}\subseteq{\cal X}$ {[}Lemma \ref{th:uspropfs-1}(iii){]}
 is bounded. Therefore, (\ref{eq:distconv}) holds true and $\{\mathbf{x}^{\nu}\}$
has a limit point in $\mathcal{X}$. By Theorem \ref{th:conver-1}(b)
and (\ref{eq:distconv}), statement (b) of the theorem
follows readily. Finally, by (\ref{eq:desclem}), $U(\mathbf{x}^{\nu})$ is a decreasing sequence: hence, no limit point of $\{\mathbf{x}^{\nu}\}$ can be a local maximum of
$U$.\smallskip{}

\noindent\textbf{Proof of statement (c).} Invoking Theorem \ref{th:conver-1}(a),
it is sufficient to show that (\ref{eq:Th_preliminary_subsequence_convergence})
in Theorem \ref{th:conver-1} is satisfied. Following the same steps
as in the proof of statement (b), by \eqref{eq:desclem} and $\gamma^{\nu}\,\to\,0$,
for $\nu\ge\bar{\nu}$ sufficiently large, there exists
a positive constant $\zeta$ such that:
\begin{equation}
U({\bold x}^{\nu+1})\le U({\bold x}^{\nu})-\gamma^{\nu}\zeta\|\hat{\bold x}({\bold x}^{\nu})-{\bold x}^{\nu}\|^{2},\label{eq:descbeta}
\end{equation}
which, again, by A4, leads to
\begin{equation}
\lim_{\nu\rightarrow\infty}\sum_{t=\bar{\nu}}^{\nu}\gamma^{t}\|\hat{\bold x}({\bold x}^{t})-{\bold x}^{t}\|^{2}<+\infty.\label{eq:summable_series}
\end{equation}
The desired result (\ref{eq:Th_preliminary_subsequence_convergence})
follows from (\ref{eq:summable_series}) and $\sum_{\nu=0}^{\infty}\gamma^{\nu}=+\infty$.
Similarly to the previous case, by (\ref{eq:summable_series}), eventually $U(\mathbf{x}^{\nu})$ is a decreasing sequence: thus, no limit point of $\{\mathbf{x}^{\nu}\}$ can be a local maximum of
$U$.

Suppose now that Assumption 4 holds. By Theorem \ref{th:conver-1}(b)
it is sufficient to prove that (\ref{eq:Th_preliminary_sequence_convergence})
holds true. For notational simplicity, we set $\Delta\hat{\bold x}({\bold x}^{\nu})\triangleq\hat{\bold x}({\bold x}^{\nu})-{\bold x}^{\nu}$.
We already proved that $\liminf_{\nu}\|\Delta\hat{\bold x}({\bold x}^{\nu})\|=0$;
therefore, (\ref{eq:Th_preliminary_sequence_convergence}) holds if $\limsup_{\nu}\|\Delta\hat{\bold x}({\bold x}^{\nu})\|=0$, as stated next.

First of all, note that, by Assumption 4, Lemma \ref{th:basicarg}
and, by consequence, Lemma \ref{th:uspropofupgpre} hold true; therefore,
there exists  $\tilde{\rho}>0$ such that (cf. Lemma \ref{th:uspropofupgpre})
\begin{equation}
\|\hat{\bold x}({\bold x}^{\nu})-\hat{\bold x}({\bold x}^{t})\|\le\eta_{\rho}\|{\bold x}^{\nu}-{\bold x}^{t}\|+\theta_{\rho}\|{\bold x}^{\nu}-{\bold x}^{t}\|^{\frac{1}{2}},\label{eq:hol_theo}
\end{equation}
for any $\nu,t\geq1$ and $\rho\in(0,\tilde{\rho}]$, with $\eta_{\rho}$
and $\theta_{\rho}$ defined in (\ref{eq:def_eta_theta}) (cf. Lemma
\ref{th:uspropofupgpre}).

Suppose by contradiction that $\limsup_{\nu}\|\Delta\hat{\bold x}({\bold x}^{\nu})\|>0$.
Then, there exists  $\delta>0$ such that $\begin{array}{l}
\|\Delta\hat{\bold x}({\bold x}^{\nu})\|>2\delta+\sqrt{\delta/2}\end{array}$ for infinitely many $\nu$, and also $\|\Delta\hat{\bold x}({\bold x}^{\nu})\|<\delta+\sqrt{\delta/2}$
for infinitely many $\nu$. Thus, there exists an infinite subset
of indices ${\cal N}$ such that, for each $\nu\in{\cal N}$ and
some $i_{\nu}>\nu$, the following hold:
\begin{equation}
\|\Delta\hat{\bold x}({\bold x}^{\nu})\|<\delta+\sqrt{\delta/2},\hspace{6pt}\|\Delta\hat{\bold x}({\bold x}^{i_{\nu}})\|>2\delta+\sqrt{\delta/2}
\label{eq:con1f}
\end{equation}
and, in case $i_{\nu} > \nu+1$,
\begin{equation}
\delta+\sqrt{\delta/2}\le\|\Delta\hat{\bold x}({\bold x}^{j})\|\le2\delta+\sqrt{\delta/2},\hspace{6pt}\nu<j<i_{\nu}.\label{eq:con2}
\end{equation}
Hence, for all $\nu\in{\cal N}$, we can write
\begin{equation}
\begin{array}{lll}
\delta & < & \|\Delta\hat{\bold x}({\bold x}^{i_{\nu}})\|-\|\Delta\hat{\bold x}({\bold x}^{\nu})\|\\[5pt]
 & \leq & \|\hat{\bold x}({\bold x}^{i_{\nu}})-\hat{\bold x}({\bold x}^{\nu})\|+\|{\bold x}^{i_{\nu}}-{\bold x}^{\nu}\|\\[5pt]
 & \overset{(a)}{\le} & (1+\eta_{\rho})\|{\bold x}^{i_{\nu}}-{\bold x}^{\nu}\|+\theta_{\rho}\|{\bold x}^{i_{\nu}}-{\bold x}^{\nu}\|^{\frac{1}{2}}\\[5pt]
 & \overset{(b)}{\le} & (1+\eta_{\rho})\,\left(2\delta+\sqrt{\delta/2}\right)\,\sum_{t=\nu}^{i_{\nu}-1}\gamma^{t}\\[5pt]
 &  & +\theta_{\rho}\left(2\delta+\sqrt{\delta/2}\right)^{\frac{1}{2}}\left(\sum_{t=\nu}^{i_{\nu}-1}\gamma^{t}\right)^{\frac{1}{2}},
\end{array}\label{eq:ineqser}
\end{equation}
where (a) is due to (\ref{eq:hol_theo}) and (b) comes from the
triangle inequality and the updating rule of the algorithm. It follows
from \eqref{eq:con1f} and \eqref{eq:ineqser} that
\begin{equation}
\begin{array}{ll}
\liminf_{\nu}\Big[(1 & \hspace{-8pt}+\eta_{\rho})\left(2\delta+\sqrt{\delta/2}\right)\sum_{t=\nu}^{i_{\nu}-1}\gamma^{t}\\
 & +\hspace{1pt}\theta_{\rho}\left(2\delta+\sqrt{\delta/2}\right)^{\frac{1}{2}}\left(\sum_{t=\nu}^{i_{\nu}-1}\gamma^{t}\right)^{\frac{1}{2}}\Big]>0.
\end{array}\label{eq:absu}
\end{equation}
We now prove that
$\|\Delta\hat{\bold x}({\bold x}^{\nu})\|\ge\delta/2$ for sufficiently
large $\nu\in{\cal N}$. Reasoning as in \eqref{eq:ineqser}, we have
\begin{equation}
\begin{array}{ll}
\|\Delta\hat{\bold x} & \hspace{-10pt}({\bold x}^{\nu+1})\|-\|\Delta\hat{\bold x}({\bold x}^{\nu})\|\le\\[5pt]
 & \hfill(1+\eta_{\rho})\|{\bold x}^{\nu+1}-{\bold x}^{\nu}\|+\theta_{\rho}\|{\bold x}^{\nu+1}-{\bold x}^{\nu}\|^{\frac{1}{2}}\\[5pt]
 & \le(1+\eta_{\rho})\gamma^{\nu}\|\Delta\hat{\bold x}({\bold x}^{\nu})\|+\theta_{\rho}(\gamma^{\nu})^{1/2}\|\Delta\hat{\bold x}({\bold x}^{\nu})\|^{\frac{1}{2}},
\end{array}\label{eq:ineqser2}
\end{equation}
for any given $\nu$. For large $\nu\in{\cal N}$, so that $(1+\eta_{\rho})\gamma^{\nu}\delta/2+\theta_{\rho} (\gamma^\nu \delta/2)^{\frac{1}{2}}<\delta/2+\sqrt{\delta/2}$, suppose by contradiction that $\|\Delta\hat{\bx}({\bx}^{\nu})\|<\delta/2$; this would give $\|\Delta\hat{\bx}({\bx}^{\nu+1})\|<\delta+\sqrt{\delta/2}$ and condition \eqref{eq:con2} (or, in case, \eqref{eq:con1f}) would be violated. Then, it must be
\begin{equation}
\|\Delta\hat{\bold x}({\bold x}^{\nu})\|\ge\delta/2.\label{eq:dbound}
\end{equation}
Using (\ref{eq:dbound}), we can show now that \eqref{eq:absu} is
in contradiction with the convergence of $\{U({\bold x}^{\nu})\}$.
By \eqref{eq:descbeta}, (possibly over a subsequence) for sufficiently large $\nu\in{\cal N}$, we have
\begin{equation}
\begin{array}{rcl}
U({\bold x}^{i_{\nu}}) & \le & U({\bold x}^{\nu})-\zeta\sum_{t=\nu}^{i_{\nu}-1}\gamma^{t}\|\Delta\hat{\bold x}({\bold x}^{t})\|^{2}\\[5pt]
 & < & U({\bold x}^{\nu})-\zeta\frac{\delta^{2}}{4}\sum_{t=\nu}^{i_{\nu}-1}\gamma^{t},
\end{array}\label{eq:fin}
\end{equation}
where, in the last inequality, we have used \eqref{eq:con2} (or, in case, \eqref{eq:con1f}) and \eqref{eq:dbound}.
Thus, since ${U({\bold x}^{\nu})}$ converges, \eqref{eq:fin} implies
$\lim_{\nu\in{\cal N}}\sum_{t=\nu}^{i_{\nu}-1}\gamma^{t}=0$, in contradiction with \eqref{eq:absu}.

\begin{remark}\label{rem:btone}
As we already mentioned in subsection \ref{subsec:gamma},
in \cite{BeckBenTTetr10} it is shown that, in the specific case of a {\em strongly convex} $U$,  $\tilde U =U$, and ${\cal K} = \Re^n$,  one can choose
$\gamma^\nu =1$ at every iteration and prove the stationarity of every limit point of the sequence generated by Algorithm 1
(assuming regularity).  For completeness we sketch how this result can be readily   obtained using our framework (and actually slightly improved on by also considering the case in which ${\cal K}$ is not necessarily $\Re^n$). The proof is based on Theorem
\ref{th:conver-1}(b) and a result in  \cite{BeckBenTTetr10}.
By Theorem \ref{th:conver-1}(b), it is enough to show that
(\ref{eq:Th_preliminary_sequence_convergence}) holds. But (\ref{eq:Th_preliminary_sequence_convergence}) does indeed hold because of
the strong convexity of $U$, as shown  at the beginning of Proposition 3.2 in \cite{BeckBenTTetr10}. Note that the strong convexity of $U$ plays here a fundamental role and that, once we remove this restrictive
assumption, things get considerably more difficult, as clearly shown by the complexity of the proof of Theorem \ref{th:conver}.
\end{remark}\vspace{-0.1cm}

\subsection{Proof of Lemma \ref{prop:suf cond inner conv}\label{sub:Proof-of-Lemma_Lip}}

\noindent (a) It is a consequence of Danskin's theorem \cite[Prop. A.43]{Bertsekas_Book-Parallel-Comp}.\\
(b) The statement follows from the uniform Lipschitz continuity of
$\widehat{{\mathbf{x}}}(\bullet;\mathbf{x}^{\nu})$ on $\mathbb{{R}}_{+}^{m}$
with constant $L_{\nabla d}$, which is proved next. For notational
simplicity, let us write $\widehat{{\mathbf{x}}}_{\boldsymbol{\lambda}}\triangleq\widehat{{\mathbf{x}}}(\boldsymbol{\lambda};\mathbf{x}^{\nu})$
and $\widehat{{\mathbf{x}}}_{\boldsymbol{\lambda}^{'}}\triangleq\widehat{{\mathbf{x}}}(\boldsymbol{\boldsymbol{\lambda}}^{'};\mathbf{x}^{\nu})$.
Defining $\mathcal{L}(\mathbf{x},\boldsymbol{{\lambda}})\triangleq\sum_{i=1}^{I}\left(\tilde{U}_{i}(\bx_{i};\bx^{\nu})+\boldsymbol{\lambda}^{T}\mathbf{\tilde{\mathbf{g}}}^{i}(\bx_i;\bx^{\nu})\right)$,
we have, by the minimum principle,
\vskip-0.3cm 

\[
\begin{array}{lll}
\left(\widehat{{\mathbf{x}}}_{\boldsymbol{\lambda}^{'}}-\widehat{{\mathbf{x}}}_{\boldsymbol{\lambda}}\right)^{T}\nabla_{\mathbf{x}}\mathcal{L}\left(\mathbf{\widehat{{\mathbf{x}}}_{\boldsymbol{\lambda}}},\boldsymbol{{\lambda}}\right) & \geq & 0\\
\left(\widehat{{\mathbf{x}}}_{\boldsymbol{\lambda}}-\widehat{{\mathbf{x}}}_{\boldsymbol{\lambda}^{'}}\right)^{T}\nabla_{\mathbf{x}}\mathcal{L}\left(\mathbf{\widehat{{\mathbf{x}}}_{\boldsymbol{\lambda}^{'}}},\boldsymbol{\lambda}^{'}\right) & \geq & 0.
\end{array}
\]
\vskip-0.1cm 

\noindent
Adding the two inequalities above and summing and subtracting $\nabla_{\mathbf{x}}\mathcal{L}\left(\mathbf{\widehat{{\mathbf{x}}}_{\boldsymbol{\lambda}}},\boldsymbol{{\lambda}}^{'}\right)$,
we obtain
\vskip-0.3cm 

\begin{equation}
\begin{array}{l}
c_{\tilde{U}}\cdot\|\widehat{{\mathbf{x}}}_{\boldsymbol{\lambda}}-\widehat{{\mathbf{x}}}_{\boldsymbol{\lambda}^{'}}\|^{2}\smallskip\\
\quad\leq\left(\widehat{{\mathbf{x}}}_{\boldsymbol{\lambda}^{'}}-\widehat{{\mathbf{x}}}_{\boldsymbol{\lambda}}\right)^{T}\left[\nabla_{\mathbf{x}}\mathcal{L}\left(\mathbf{\widehat{{\mathbf{x}}}_{\boldsymbol{\lambda}}},\boldsymbol{{\lambda}}\right)-\nabla_{\mathbf{x}}\mathcal{L}\left(\mathbf{\widehat{{\mathbf{x}}}_{\boldsymbol{\lambda}}},\boldsymbol{{\lambda}}^{'}\right)\right]\smallskip\\
\quad = \left(\widehat{{\mathbf{x}}}_{\boldsymbol{\lambda}}-\widehat{{\mathbf{x}}}_{\boldsymbol{\lambda}^{'}}\right)^{T}\left[\nabla_{\mathbf{x}}\mathcal{L}\left(\mathbf{\widehat{{\mathbf{x}}}_{\boldsymbol{\lambda}}},\boldsymbol{{\lambda}}^{'}\right)-\nabla_{\mathbf{x}}\mathcal{L}\left(\mathbf{\widehat{{\mathbf{x}}}_{\boldsymbol{\lambda}}},\boldsymbol{{\lambda}}\right)\right],
\end{array}\label{eq:Lip_eq1}
\end{equation}
where, in the first inequality, we used the uniform strong convexity
of $\mathcal{L}\left(\bullet,\boldsymbol{{\lambda}}^{'}\right)$.
Hence, we have
\vskip-0.3cm 

\begin{equation}
\begin{array}{rll}
c_{\tilde{U}}\cdot\|\widehat{{\mathbf{x}}}_{\boldsymbol{\lambda}}-\widehat{{\mathbf{x}}}_{\boldsymbol{\lambda}^{'}}\| & \leq & {\displaystyle {\sum_{j=1}^{m}}}\left|\lambda_{j}^{'}-\lambda_{j}\right|\left\Vert \nabla_{\mathbf{x}}\tilde{{g}}_{j}\left(\mathbf{\widehat{{\mathbf{x}}}_{\boldsymbol{\lambda}}};\mathbf{x^{\nu}}\right)\right\Vert \smallskip\\
 & \stackrel{(a)}{\leq} & {\displaystyle {\sum_{j=1}^{m}}}\left|\lambda_{j}^{'}-\lambda_{j}\right|\, L_{\tilde{{g}}}=\left\Vert \boldsymbol{{\lambda}}^{'}-\mathbf{\boldsymbol{{\lambda}}}\right\Vert _{1}\cdot L_{\tilde{{g}}}\\
 & \leq & L_{\tilde{{g}}}\,\sqrt{{m}}\,\left\Vert \boldsymbol{{\lambda}}^{'}-\mathbf{\boldsymbol{{\lambda}}}\right\Vert _{2},
\end{array}\label{eq:Lip_eq1-1}
\end{equation}
\vskip-0.2cm

\noindent
where (a) follows from the uniform Lipschitz continuity of $\tilde{{\mathbf{g}}}$.
The inequality above proves the Lipschitz property of $\widehat{{\mathbf{x}}}(\bullet;\mathbf{x}^{\nu})$.

\bibliographystyle{IEEEtran}
\bibliography{scutari_refs,Surbib}

\end{document}